\documentclass[journal,10pt,twocolumn]{IEEEtran}

\ifCLASSINFOpdf
\else
\fi

\usepackage{microtype}
\usepackage{graphicx}
\usepackage{subfigure}
\usepackage{booktabs,pifont} 

\usepackage{empheq}

\DeclarePairedDelimiter\abs{\lvert}{\rvert}
\usepackage{pgfplots}
\pgfplotsset{compat=1.12}
\usetikzlibrary{shapes,arrows}
\usetikzlibrary{positioning}
\usepackage{tikz}
\usetikzlibrary{positioning,chains,fit,shapes,calc}
\usetikzlibrary{patterns}
\usepackage{amsmath,bm,times}
\usetikzlibrary{calc}
\usepackage{dsfont}
\usepackage{cuted}
\usepackage{caption}

\usepackage{mathtools}
\DeclarePairedDelimiter{\ceil}{\lceil}{\rceil}
\DeclarePairedDelimiter{\floor}{\lfloor}{\rfloor}
\usepackage{amsthm}
\usepackage{comment}
\usepackage{enumitem}
\usepackage{arydshln}
\usepackage{cite}
\usepackage{multirow}
\usepackage{calc}
\usepackage{blkarray}
\usepackage{url}
\theoremstyle{definition}
\newtheorem{theorem}{Theorem}

\newtheorem{lemma}{Lemma}
\newtheorem{claim}{Claim}

\newtheorem{corollary}{Corollary}
\newtheorem{example}{Example}
\newtheorem{remark}{Remark}
\newtheorem{definition}{Definition}
\usepackage{graphicx}
\usepackage{dblfloatfix}
\usetikzlibrary{decorations.pathreplacing}
\usepackage{blindtext, graphicx, amsfonts,
	amssymb,multirow,epstopdf}
\def\BibTeX{{\rm B\kern-.05em{\sc i\kern-.025em b}\kern-.08em
    T\kern-.1667em\lower.7ex\hbox{E}\kern-.125emX}}

\makeatletter
\renewcommand*\env@matrix[1][*\c@MaxMatrixCols c]{%
  \hskip -\arraycolsep
  \let\@ifnextchar\new@ifnextchar
  \array{#1}}
\makeatother
\usepackage[linesnumbered,ruled]{algorithm2e}

\newcommand{\cmark}{\ding{51}}%
\newcommand{\xmark}{\ding{55}}
\newcommand{\threevdots}{%
  \vbox{\baselineskip1ex\lineskiplimit0pt%
  \hbox{.}\hbox{.}\hbox{.}}}

\newcommand\bovermat[2]{%
    \makebox[0pt][l]{$\smash{\overbrace{\phantom{%
                    \begin{matrix}#2\end{matrix}}}^{\text{#1}}}$}#2}

\newcommand\undermat[2]{
  \makebox[0pt][l]{$\smash{\underbrace{\phantom{%
    \begin{matrix}#2\end{matrix}}}_{\text{$#1$}}}$}#2}
\newcommand{\calB}{\mathcal{B}}
\newcommand{\calC}{\mathcal{C}}

\newcommand{\calG}{\mathcal{G}}

\newcommand{\calN}{\mathcal{N}}
\newcommand{\bfU}{\mathbf{U}}
\newcommand{\bfC}{\mathbf{C}}
\newcommand{\bfb}{\mathbf{b}}
\newcommand{\bfy}{\mathbf{y}}

\newcommand{\bfV}{\mathbf{V}}

\newcommand{\bfY}{\mathbf{Y}}
\newcommand{\bfW}{\mathbf{W}}
\newcommand{\bfA}{\mathbf{A}}

\newcommand{\bfG}{\mathbf{G}}

\newcommand{\bfu}{\mathbf{u}}

\newcommand{\bfx}{\mathbf{x}}

\newcommand{\bfc}{\mathbf{c}}
\newcommand{\bfB}{\mathbf{B}}

\newcommand{\bfz}{\mathbf{z}}

\newcommand{\bfR}{\mathbf{R}}

\newcommand{\aditya}[1]{\marginpar{+}{\bf Aditya's remark}: {\em #1}}

\begin{document}

\title{A Unified Treatment of Partial Stragglers and Sparse Matrices in Coded Matrix Computation}

\definecolor{mygr}{rgb}{0.6,0.4,0.0}
\definecolor{my1color}{rgb}{0.6,0.4,0.0}
\definecolor{mycolor1}{rgb}{0.00000,0.44700,0.74100}%
\definecolor{mycolor2}{rgb}{0.85000,0.32500,0.09800}%
\definecolor{mycolor3}{rgb}{0.45000,0.62500,0.19800}%
\tikzset{
block/.style    = {draw, thick, rectangle, minimum height = 2em, minimum width = 2em},
sum/.style      = {draw, circle, node distance = 1cm},
sum1/.style      = {draw, circle, minimum size = 1.1 cm},
input/.style    = {coordinate},
output/.style   = {coordinate},
}

\author{\IEEEauthorblockN{Anindya Bijoy Das and Aditya Ramamoorthy} 


\thanks{This work was supported in part by the National Science Foundation (NSF) under grants CCF-1910840 and CCF-2115200. The material in this work has appeared in part at the 2021 Information Theory Workshop (ITW), Kanazawa, Japan and the 2022 IEEE International Symposium on Information Theory (ISIT), Aalto University in Espoo, Finland.

Anindya Bijoy Das (das207@purdue.edu) was with the Dept. of Electrical and Computer Engineering at Iowa State University. He is now with the Dept. of Electrical and Computer Engineering, Purdue University, West Lafayette, IN 47907 USA.

Aditya Ramamoorthy (adityar@iastate.edu) is with the Dept. of Electrical and Computer Engineering, Iowa State University, Ames, IA 50011 USA.}
}

\IEEEtitleabstractindextext{%

\begin{abstract}
The overall execution time of distributed matrix computations is often dominated by slow worker nodes (stragglers) within the computation clusters. Recently, coding-theoretic techniques have been utilized to mitigate the effect of stragglers where worker nodes are assigned the job of processing encoded submatrices of the original matrices. In many machine learning or optimization problems the relevant matrices are often sparse. Several prior coded computation methods operate with dense linear combinations of the original submatrices; this can significantly increase the worker node computation times and consequently the overall job execution time. Moreover, several existing techniques treat the stragglers as failures (erasures) and discard their computations. In this work, we present a coding approach which operates with limited encoding of the original submatrices and utilizes the partial computations done by the slower workers. While our scheme can continue to have the optimal threshold of prior work, it also allows us to trade off the straggler resilience with the worker computation speed for sparse input matrices. Extensive numerical experiments done over cloud platforms confirm that the proposed approach enhances the speed of the worker computations (and thus the whole process) significantly.
\end{abstract}
}
\maketitle
\IEEEdisplaynontitleabstractindextext
\IEEEpeerreviewmaketitle
\section{Introduction}
\label{sec:intro}
Matrix computations are an indispensable part of several machine learning and optimization problems. The large scale dimensions of the matrices in these problems necessitates the usage of distributed computation where the whole job is subdivided into smaller tasks and assigned to multiple worker nodes. In these systems, the overall job execution time can be dominated by slower (or failed) worker nodes, which are referred to as stragglers. Recently, a number of coding theory techniques \cite{lee2018speeding,dutta2016short,yu2017polynomial, yu2020straggler, tandon2017gradient, 8849395, 8849468, 8919859, ramamoorthy2019numerically, das2019random, 9513242} have been proposed to mitigate the effect of stragglers for matrix-vector and matrix-matrix multiplications (see \cite{ramamoorthyDTMag20} for a tutorial overview). For example, \cite{lee2018speeding} proposes to compute $\bfA^T \bfx$, where $\bfA \in \mathbb{R}^{t \times r}$ and $\bfx \in \mathbb{R}^{t}$, in a distributed fashion by partitioning the matrix $\bfA$ into two block-columns as $\bfA = [\bfA_0 ~|~ \bfA_1]$, and assigning the job of computing $\bfA^T_0 \bfx$, $\bfA^T_1 \bfx$ and $\left(\bfA_0+\bfA_1\right)^T \bfx$, respectively, to three different workers. Thus, we can recover $\bfA^T \bfx$ if {\it any} two out of three workers return their results. This implies that the system is resilient to one straggler, while each of those workers has effectively half of the overall computational load. In general, we define the recovery threshold as the minimum number of workers ($\tau$) that need to finish their jobs such that the result $\bfA^T \bfx$ (for matrix-vector multiplication), or $\bfA^T \bfB$ (for matrix-matrix multiplication; where $\bfB \in \mathbb{R}^{t \times w}$) can be recovered from any subset of $\tau$ worker nodes.

While the recovery threshold is an important metric for coded computation, there are other important issues that also need to be considered.
First, in many of the practical examples, the matrices $\bfA$ and/or $\bfB$ can be sparse. If we linearly combine  $m$ submatrices of $\bfA$, then the density of the non-zero entries in the encoded matrices can be up to $m$-times higher than the density of $\bfA$; the actual value will depend on the underlying sparsity pattern. This can result in a large increase in the worker node computation time. As noted in \cite{wang2018coded}, the overall job execution time can actually go up rather than down. Thus, designing schemes that combine relatively few submatrices while continuing to have a good threshold is important. Second, much of prior work (see \cite{wang2018coded, das2020coded, kiani2018exploitation, mallick2018rateless, 9252114, 8683267, 8849451, hasirciouglu2020bivariate} for some exceptions) treats stragglers as erasures, and thus discards the computations done by the slower workers. But a slower worker may not be a useless worker and efficiently utilizing the partial computations done by the slower workers is of interest.

In this work, we propose an approach for both distributed matrix-vector and matrix-matrix multiplications which makes progress on both of the issues mentioned above. In our approach, many of the assigned submatrices within a worker node are uncoded; this makes the worker computations significantly faster. Moreover, our proposed approach can exploit the partial computations done by the slower workers. We emphasize that our approach continues to enjoy the optimal recovery threshold (see \cite{yu2017polynomial}) for storage fractions of the form $1/k_A$ and $1/k_B$ (for the worker nodes). Furthermore, we also present approaches to trade off the straggler resilience with number of submatrices of $\bfA$ and $\bfB$ that are encoded.

This paper is organized as follows. In Section \ref{sec:backrelsoc}, we discuss the problem setup, related works and summarize our contributions. Then, in Section \ref{sec:matmat}, we present the details of our scheme and discuss our results about straggler resilience, suitability for sparse input matrices and utilization of partial computation. Next, in Section \ref{sec:numexp}, we show the results of extensive numerical experiments and compare the performance of our proposed method with other methods. Finally, Section \ref{sec:conclusion} concludes the paper with a discussion of possible future directions.

\section{Problem Setup, Related Work and Summary of Contributions}
\label{sec:backrelsoc}
In this work by matrix-computations we refer to either matrix-vector multiplication or matrix-matrix multiplication. In the case of matrix-vector multiplication, we typically partition matrix $\bfA$ into submatrices, generate encoded submatrices and distribute a certain number of these encoded submatrices and the vector $\bfx$ to $n$ worker nodes depending on their storage capacities. Each worker node computes (in a specified order) the product of its assigned submatrices and the vector $\bfx$. Each time a product is computed, it is communicated to the central node. The central node is responsible for recovering $\bfA^T \bfx$ once enough products have been obtained from the workers. 

Matrix-matrix multiplication is a more general and challenging problem. Here, we first perform a block decomposition of matrices $\bfA$ and $\bfB$ with sizes $q \times u$ and $q \times v$ respectively. The $(i,j)$-th block of $\bfA$ is denoted $\bfA_{i,j}, 0 \leq i \leq q-1, 0 \leq j \leq u-1$ (similar notation holds for the blocks of $\bfB$). The central node carries out encoding on both the $\bfA_{i,j}$ submatrices and the $\bfB_{i,j}$ submatrices. In particular, although it calculates scalar linear combinations of the submatrices, it is not responsible for any of the computationally intensive matrix operations. Following this, it sends the coded submatrices of $\bfA$ and $\bfB$ to the workers. The worker nodes compute the corresponding pair-wise products (either all or some subset thereof) of the submatrices assigned to them in a specified order and send the results back to the central node which performs appropriate decoding to recover $\bfA^T \bfB$.  

\subsection{Problem Setup}
We assume that the system has $n$ workers each of which can store the equivalent of $\gamma_A = \frac{1}{k_A}$ and $\gamma_B = \frac{1}{k_B}$ fractions of matrices $\bfA$ and $\bfB$. It is assumed that some of the worker nodes will fail or will be too slow (which is often the case in real-life clusters); the number of such nodes is assumed to be $s$ (or less). 

\begin{definition}
\label{def:recthresh}
We define the recovery threshold as the minimum number of workers ($\tau$) that need to finish their jobs such that the result $\bfA^T \bfx$ (for matrix-vector multiplication, where $\bfA \in \mathbb{R}^{t \times r}$ and $\bfx \in \mathbb{R}^{t}$), or $\bfA^T \bfB$ (for matrix-matrix multiplication; where $\bfA \in \mathbb{R}^{t \times r}$ and $\bfB \in \mathbb{R}^{t \times w}$) can be recovered from any subset of $\tau$ worker nodes.
\end{definition}
 The recovery threshold of a scheme is said to be {\it optimal} if it is the lowest possible given the storage constraints of the worker nodes.  
 
In our approach, we partition matrix $\bfA$ into $\Delta_A = \textrm{LCM}(n, k_A)$ submatrices (block-columns) as $\bfA_0, \bfA_1, \bfA_2, \dots, \bfA_{\Delta_A - 1}$ where LCM indicates the least common multiple. We also partition matrix $\bfB$ into $\Delta_B = k_B$ submatrices (block-columns) as $\bfB_0, \bfB_1, \bfB_2, \dots, \bfB_{\Delta_B - 1}$ and set $\Delta = \Delta_A \Delta_B$. We denote the number of assigned submatrices from $\bfA$ and $\bfB$ to any worker as $\ell_A$ and $\ell_B$ respectively, so $\ell_A = \frac{\Delta_A}{k_A}$ and $\ell_B = \frac{\Delta_B}{k_B} = 1$. Any worker will compute all pairwise block-products, thus the worker will be responsible for computing $\ell = \ell_A \ell_B = \ell_A$ block-products. 

We say that any submatrix $\bfA_i$, for $i = 0, 1, \dots, \Delta_A-1$, appears within a worker node as an uncoded block if $\bfA_i$ is assigned to that worker as an uncoded submatrix. Similarly, $\bfA_i$ is said to appear within a worker node in a coded block if a random linear combination of some submatrices including $\bfA_i$ is assigned to that worker. At this point, we define the ``weight'' of the encoded submatrices as it will serve as an important metric for working with sparse input matrices.  

\begin{definition}
\label{def:weight}
We define the ``weight'' of the encoding as the number of submatrices that are linearly combined to arrive at the encoded submatrix.
\end{definition}

If the ``input'' submatrices $\bfA$ and/or $\bfB$ are sparse, the encoded submatrices will be denser and the density is proportional to the weight of the encoding. Computing the product of two dense matrices is more computationally expensive than computing the product of two sparse matrices.
Thus, low weight encodings are desirable to enhance the speed of overall computation in case  of sparse matrices.


In each worker node there are locations numbered $0, 1, \dots, \ell-1$ where $0$ indicates the top location and $\ell-1$ the bottom location. The worker node starts processing the assigned submatrix product at the top (location 0) and then proceeds downwards to location $\ell-1$. For this system, if the central node can decode $\bfA^T \bfB$ from {\it any} $Q$ block products (respecting the top-to-bottom computation order), we say that the scheme has the corresponding $Q/\Delta$ value. A smaller $Q/\Delta$ value of a system indicates that the system can utilize the partial computations of the slower workers more efficiently than a system with higher $Q/\Delta$ value. It is to be noted that in our problem there are a total of $\Delta$ submatrix products to be recovered. Hence, $Q/\Delta \geq 1$.

We shall use the terminology of symbols and submatrices (or submatrix-products) interchangeably at several places.
\subsection{Discussion of Related Work}
Several coded computation approaches \cite{lee2018speeding,dutta2016short,yu2017polynomial, yu2020straggler, tandon2017gradient, 8849395, 8849468, 8919859, ramamoorthy2019numerically, das2019random, 9513242} have been introduced in the literature of distributed matrix multiplication in recent years. Many of these ideas are presented in a tutorial fashion in \cite{ramamoorthyDTMag20}. We compare the properties of different coded matrix-multiplication schemes in Table \ref{compare}. Moreover in Table \ref{notation}, we outline the notations used in this paper. 

\begin{table*}[t]
\caption{{\small Comparison with existing works on distributed matrix-matrix multiplications. We did not add the works of \cite{mallick2018rateless, 8849395, keshtkarjahromi2018dynamic, c3les} in this table since they are applicable for matrix-vector multiplication only.}} 
\label{compare}
\begin{center}
\begin{small}
\begin{sc}
\begin{tabular}{c c c c c}
\hline
\toprule
\multirow{2}{1 cm}{Codes} & Optimal & Numerical & Partial & Sparsely\\
  & Threshold? & Stability? & Computation? & Coded?\\
 \midrule
Repetition Codes & \xmark & \cmark &  \xmark & \cmark\\ \hline
Prod. Codes \cite{lee2017high}, Factored Codes \cite{9513242}   &\xmark & \cmark  & \xmark & \xmark\\ \hline
Polynomial Codes \cite{yu2017polynomial}  & \cmark & \xmark & \xmark & \xmark\\ \hline
Biv. Hermitian Poly. Code \cite{hasirciouglu2020bivariate}  & \cmark & \xmark & \cmark & \xmark\\ \hline
OrthoPoly \cite{8849468}, RKRP code\cite{8919859} & \multirow{2}{0.3 cm}{\cmark} & \multirow{2}{0.3 cm}{\cmark}  & \multirow{2}{0.3 cm}{\xmark} & \multirow{2}{0.3 cm}{\xmark}\\
Conv. Code \cite{das2019random}, Circ. \& Rot. Mat. \cite{ramamoorthy2019numerically} & &  & & \\ \hline
$\beta$-level Coding \cite{das2020coded} & \xmark & \cmark & \cmark & \cmark\\ \hline
SCS Optimal Scheme \cite{das2020coded} & \cmark & \cmark & \cmark & \cmark\\ \hline
{\bf Proposed Scheme} & \cmark & \cmark & \cmark & \cmark\\

\bottomrule
\end{tabular}
\end{sc}
\end{small}
\end{center}
\end{table*}%

\begin{table*}[t]
\caption{{\small Notation Table}} 
\label{notation}
\begin{center}
\begin{small}
\begin{sc}
\begin{tabular}{c c c}
\hline
\toprule
notation & definition & description\\
 \midrule
$\gamma_A, \gamma_B$ & $\;\;\;$ storage fraction for $\bfA$ and $\bfB$, respectively & $\gamma_A = \frac{1}{k_A}, \gamma_B = \frac{1}{k_B}$ \\  \hline
 $n$ & number of total worker nodes & $n \geq k_A k_B$\\ \hline
 $s_m$ & number of maximum possible full stragglers & $s_m = n - k_A k_B$ \\ \hline
\multirow{2}{1 cm}{$\Delta_A, \Delta_B$} & number of block-columns that & $\Delta_A = \textrm{LCM}(n, k_A)$ \\ 
& $\bfA$ and $\bfB$, respectively, are partitioned into &  and $\Delta_B = k_B$\\  \hline
$\Delta$ & total number of unknowns that need to be recovered & $\Delta = \Delta_A \Delta_B$ \\ \hline
$\tau$ & recovery threshold of the scheme & $\tau \geq k_A k_B$ \\ \hline
$s$ & number of stragglers & $s = n - \tau$ \\ \hline
$x$ & relaxation in number of stragglers & $x = s_m - s$ \\ \hline
$y$ & reduction of weights in coded submatrices of $\bfA$ & $y = \floor{\frac{k_A x}{s_m}}$ \\ \hline
\multirow{2}{1 cm}{$\;\;\; Q$} & number of submatrix products that have to be computed & \multirow{2}{2 cm}{$\;\;\;  \;\;\; Q \geq \Delta$} \\ 
 & in the worst case to recover the intended result &  \\ \hline
$\zeta$ & weight for the encoding of $\bfB$ & $\zeta \leq k_B$ \\ \hline
\multirow{2}{2 cm}{$\;\;\;  \;\;\; \ell_u, \ell_c$} & number of uncoded and coded submatrices of $\bfA$ & $\ell_u = \frac{\Delta}{n}$ and \\
& assigned to every worker node & $ \; \; \ell_ u + \ell_c = \Delta_A /k_A$\\
\bottomrule
\end{tabular}
\end{sc}
\end{small}
\end{center}
\end{table*}%



With storage fractions $\gamma_A = 1/k_A$ and $\gamma_B = 1/k_B$ and $q=1$, the approach in \cite{yu2017polynomial} has a threshold $\tau = k_A k_B$ which is shown to be optimal \cite{yu2017polynomial}. It proceeds by creating encoded matrix polynomials whose coefficients correspond to the blocks of the input matrices and subsequently evaluating them at distinct points. The decoding process corresponds to polynomial interpolation. Moreover, there are other variants of polynomial code-based works \cite{8849468}, \cite{hasirciouglu2020bivariate}, random code-based approaches \cite{8919859}, or convolutional code-based methods \cite{8849395,das2019random} which are also resilient to optimal number of stragglers. It should be noted that there are several other approaches \cite{wang2018coded}, \cite{mallick2018rateless}, \cite{lee2017high} which are sub-optimal in terms of straggler resilience.

There are some works \cite{dutta2016short,dutta2019optimal,yu2020straggler,TangKR19} which consider the case of $q > 1$ in the block decomposition of the input matrices. While this can reduce the recovery threshold as compared to the case of $q = 1$, it comes at the cost of increased computational load at each of the workers. Moreover the communication load between the worker node and the central node also increases.

While much of the initial work on coded computation focused on the recovery threshold, subsequent work has identified other metrics that need to be considered as well. Here we discuss several such other concerns that have been discussed in this literature.

\noindent {\bf Sparsity of the ``input'' matrices:} There are practical applications in machine learning, optimization and other areas where the underlying matrices $\bfA$ and $\bfB$ are sparse. If we want to compute the inner product of two $n$-length vectors $\mathbf{a}$ and $\bfx$ where $\mathbf{a}$ has around $\delta n$ ($0 < \delta \ll 1$) non-zero entries, it takes $\approx 2 \delta n$ floating point operations (flops) as compared to $\approx 2 n$ flops in the dense case where $\delta \approx 1$. The linear encoding in several prior approaches \cite{yu2017polynomial, 8919859, yu2020straggler, 8849468, das2019random, hasirciouglu2020bivariate}  is dense, i.e., it significantly increases the number of non-zero entries in the encoded matrices which are assigned to the worker nodes. For example, in the approach of \cite{yu2017polynomial} the encoded matrices of $\bfA$ (respectively $\bfB$) are obtained by linearly combining $k_A$ (respectively $k_B$) submatrices. This in turn can cause the worker computation time to increase by up to $k_A k_B$ times, i.e., the advantage of distributing the computation may be lost. This underscores the necessity of considering schemes where the encoding only combines a limited number of submatrices.  


\noindent {\bf Numerical stability:} Another major issue of the polynomial based approaches \cite{yu2017polynomial} is numerical instability. It should be noted that the encoding and decoding algorithms within coded computation operate over the real field, although the corresponding techniques are borrowed from classical coding theory. Unlike the finite field, the recovery from a system of equations can be quite inaccurate in the real field if the corresponding system matrix is ill-conditioned. The polynomial code approach \cite{yu2017polynomial} uses Vandermonde matrices for the encoding process which are well-recognized to be ill-conditioned \cite{Pan16}.


A number of prior works \cite{8849468, 8919859, ramamoorthy2019numerically, das2019random, 8849451,  9174440} have addressed this issue and emphasized that the worst case condition number ($\kappa_{worst}$) of the decoding matrices over all different choices of $s$ stragglers is an important metric to be optimized.  The work of \cite{8849468} presents an approach within the basis of orthogonal polynomials, and demonstrates that $\kappa_{worst}$ of their schemes is at most $O(n^{2s})$. The approach in \cite{das2019random} proposes a random convolutional code based approach, and provides a computable upper bound on $\kappa_{worst}$, while the work of \cite{ramamoorthy2019numerically} leverages the properties of rotation matrices and circulant permutation matrices to upper bound $\kappa_{worst}$ by $O(n^{s+5.5})$. In terms of numerical stability, \cite{ramamoorthy2019numerically} provides the best results in numerical experiments. The work in \cite{8919859} presents an approach where they take random linear combinations of the submatrices to generate the coded submatrices (this was also suggested in Remark 8 of \cite{yu2017polynomial}) and shows the improvement over the polynomial code approach. Some other approaches \cite{8849395}, \cite{8849451} address this issue, but they are applicable to matrix-vector multiplication only.

\noindent {\bf Partial Stragglers:} The third issue in distributed computations is that several approaches \cite{yu2017polynomial, 8849468, 8919859, ramamoorthy2019numerically, das2019random} treat stragglers as erasures; in other words they assume that no useful information can be obtained from the slower worker nodes. But some recent works \cite{das2020coded, c3les, kiani2018exploitation, hasirciouglu2020bivariate} consider that a slow worker may not be a useless worker; rather exploiting these partial computations can enhance the speed of the overall job. In these approaches, multiple jobs are assigned to each of the worker nodes, so that the central node can leverage the partial computations. This naturally leads to the $Q/\Delta$-metric discussed previously; it was introduced in \cite{c3les} and discussed in-depth in \cite{das2020coded}. 
We present a detailed comparison with \cite{das2020coded} in Section \ref{sec:numexp}.

\subsection{Summary of Contributions}

The contributions of our work can be summarized as follows. 

\begin{itemize}
    \item For a system with $n$ workers each of which can store $\gamma_A = \frac{1}{k_A}$ fraction of matrix $\bfA$ and $\gamma_B = \frac{1}{k_B}$ fraction of matrix $\bfB$, we propose a coded matrix-matrix multiplication scheme which (i) is optimal in terms of straggler resilience ($s = n - k_A k_B$); (ii) can utilize the partial computations done by the slower worker nodes; and (iii) enhances the worker computation speed when the ``input'' matrices $\bfA$ and $\bfB$ are sparse. Specifically, several of the assigned submatrices in our scheme our uncoded.
    
    \item Our work allows us to trade off the straggler resilience with the weight of the encoding scheme. If the recovery threshold is relaxed to $\tau = k_A k_B + x$, then we can further reduce the weight of the encoded $\bfA$ submatrices while ensuring that the number of uncoded $\bfA$ submatrices remains the same (as in the optimal threshold case). We show that the coded submatrices will be linear combinations of $k_A - y$ uncoded submatrices; where $y = \floor{\frac{k_A x}{s_m}}$. Thus the worker computation speed can be enhanced in comparison to the case of $\tau = k_A k_B$ when $x = 0$. 
    
    \item We provide upper and lower bounds on the value of $Q$ for our scheme. We show that for $x = 0$, the bounds are the same. Moreover, we have demonstrated several numerical examples which show that the difference between the bounds is small even when $x > 0$.
    
    \item Our theoretical results are supported by extensive numerical experiments conducted on AWS clusters. Fig. \ref{over_comp_sparse}, depicts such a comparison in terms of overall computation time required by different approaches for sparse input matrices in a system of $n = 24$ worker nodes. We have simulated the stragglers in such a way that the slower workers have one-fifth of the speed of the non-straggling nodes. From Fig. \ref{over_comp_sparse}, it can be verified that our proposed approach requires significantly less overall computation time than the dense coded approaches when the ``input'' matrices are sparse. 
    
\end{itemize}

\begin{figure*}[t]
\begin{subfigure}
\centering
\resizebox{0.48\linewidth}{!}{

\definecolor{mycolor6}{rgb}{0.92941,0.69412,0.12549}%
\definecolor{mycolor7}{rgb}{0.74902,0.00000,0.74902}%
\definecolor{mycolor8}{rgb}{0.60000,0.20000,0.00000}%

\begin{tikzpicture}
\begin{axis}[%
width=5.1in,
height=3.003in,
at={(2.6in,0.85in)},
scale only axis,
xmin=0,
xmax=6,
xlabel style={font=\color{white!15!black}, font=\huge},
xlabel={Number of slower workers},
ymin=0,
ymax=18,
ytick={0, 4, 8, 12, 16},
xtick={0,1,2,3,4,5,6},
tick label style={font=\LARGE} ,
ylabel style={font=\color{white!15!black}, font=\LARGE},
ylabel={Overall computation time (in $s$)},
axis background/.style={fill=white},
xmajorgrids={true},
ymajorgrids={true},
legend style={at={(0.03,0.5)}, nodes={scale=1}, anchor=south west, legend cell align=left, align=left, draw=white!15!black,font = \LARGE}
]

\addplot [solid, color=blue, line width=2.0pt, mark=diamond, mark options={solid, blue, scale = 3}]
  table[row sep=crcr]{%
0	3.33\\
1	3.44\\
2	3.47\\
3	3.46\\
4	3.52\\
5	14.93\\
6	14.83\\
};
\addlegendentry{Polynomial code  \cite{yu2017polynomial}}

\addplot [dashed, color=mycolor6, line width=2.0pt, mark=*, mark options={solid, mycolor6, scale = 2}]
  table[row sep=crcr]{%
0	3.38\\
1	3.42\\
2	3.48\\
3	3.39\\
4	3.63\\
5	14.95\\
6	14.70\\
};
\addlegendentry{Ortho-Poly Code \cite{8849468}}

\addplot [dashed, color=mycolor1, line width=2.0pt, mark=*, mark options={solid, mycolor1, scale = 3}]
  table[row sep=crcr]{%
0	0.88\\
1	2.85\\
2	2.84\\
3	3.04\\
4	3.06\\
5	3.04\\
6	3.29\\
};
\addlegendentry{RKRP Code\cite{8919859}}

\addplot [dotted, color=red, line width=2.0pt, mark=o, mark options={solid, red, scale = 3}]
  table[row sep=crcr]{%
0	1.15\\
1	1.82\\
2	1.85\\
3	1.90\\
4	1.94\\
5	2.12\\
6	2.40\\
};
\addlegendentry{SCS optimal Scheme \cite{das2020coded}}

\addplot [dotted, color=mycolor7, line width=2.0pt, mark=*, mark options={solid, mycolor7, scale = 3}]
  table[row sep=crcr]{%
0	0.66\\
1	0.97\\
2   0.95\\
3	1.00\\
4	1.07\\
5	1.74\\
6	2.31\\
};
\addlegendentry{Proposed Scheme}

\end{axis}
\end{tikzpicture}%
}
\end{subfigure}
\begin{subfigure}
\centering
\resizebox{0.48\linewidth}{!}{

\definecolor{mycolor6}{rgb}{0.92941,0.69412,0.12549}%
\definecolor{mycolor7}{rgb}{0.74902,0.00000,0.74902}%
\definecolor{mycolor8}{rgb}{0.60000,0.20000,0.00000}%

\begin{tikzpicture}
\begin{axis}[%
width=5.1in,
height=3.003in,
at={(2.6in,0.85in)},
scale only axis,
xmin=0,
xmax=6,
xlabel style={font=\color{white!15!black}, font=\huge},
xlabel={Number of slower workers},
ymin=0,
ymax=50,
ytick={0, 12, 24, 36, 48},
xtick={0,1,2,3,4,5,6},
tick label style={font=\LARGE} ,
ylabel style={font=\color{white!15!black}, font=\LARGE},
ylabel={Overall computation time (in $s$)},
axis background/.style={fill=white},
xmajorgrids={true},
ymajorgrids={true},
legend style={at={(0.03,0.5)}, nodes={scale=1}, anchor=south west, legend cell align=left, align=left, draw=white!15!black,font = \LARGE}
]

\addplot [solid, color=blue, line width=2.0pt, mark=diamond, mark options={solid, blue, scale = 3}]
  table[row sep=crcr]{%
0	8.64\\
1	8.86\\
2	8.14\\
3	8.13\\
4	10.12\\
5	40.51\\
6	39.80\\
};
\addlegendentry{Polynomial code  \cite{yu2017polynomial}}

\addplot [dashed, color=mycolor6, line width=2.0pt, mark=*, mark options={solid, mycolor6, scale = 2}]
  table[row sep=crcr]{%
0	8.20\\
1	8.26\\
2	9.34\\
3	9.61\\
4	10.11\\
5	38.73\\
6	39.47\\
};
\addlegendentry{Ortho-Poly Code \cite{8849468}}

\addplot [dashed, color=mycolor1, line width=2.0pt, mark=square, mark options={solid, mycolor1, scale = 3}]
  table[row sep=crcr]{%
0	2.26\\
1	5.46\\
2	6.28\\
3	7.57\\
4	7.81\\
5	7.85\\
6	7.27\\
};
\addlegendentry{RKRP Code\cite{8919859}}

\addplot [dotted, color=red, line width=2.0pt, mark=o, mark options={solid, red, scale = 3}]
  table[row sep=crcr]{%
0	3.16\\
1	4.51\\
2	4.64\\
3	4.80\\
4	4.93\\
5	5.13\\
6	5.88\\
};
\addlegendentry{SCS optimal Scheme \cite{das2020coded}}

\addplot [dotted, color=mycolor7, line width=2.0pt, mark=*, mark options={solid, mycolor7, scale = 3}]
  table[row sep=crcr]{%
0	2.28\\
1	3.52\\
2	3.56\\
3	3.78\\
4	3.82\\
5	5.77\\
6	6.13\\
};
\addlegendentry{Proposed Scheme}

\end{axis}
\end{tikzpicture}%
}
\end{subfigure}

\caption{\small Comparison among different coded approaches in terms of overall computation time for different number of slower worker nodes when the ``input'' matrices are $98\%$ sparse (left) or $95\%$ sparse (right). The system has $n = 24$ worker nodes each of which can store $\gamma_A = 1/4$ and $\gamma_B = 1/5$ fraction of matrices $\bfA$ and $\bfB$, respectively, so the recovery threshold, $\tau = 20$. The slower workers are simulated in such a way so that they have one-fifth of the speed of the non-straggling workers.}
\label{over_comp_sparse}
\end{figure*}
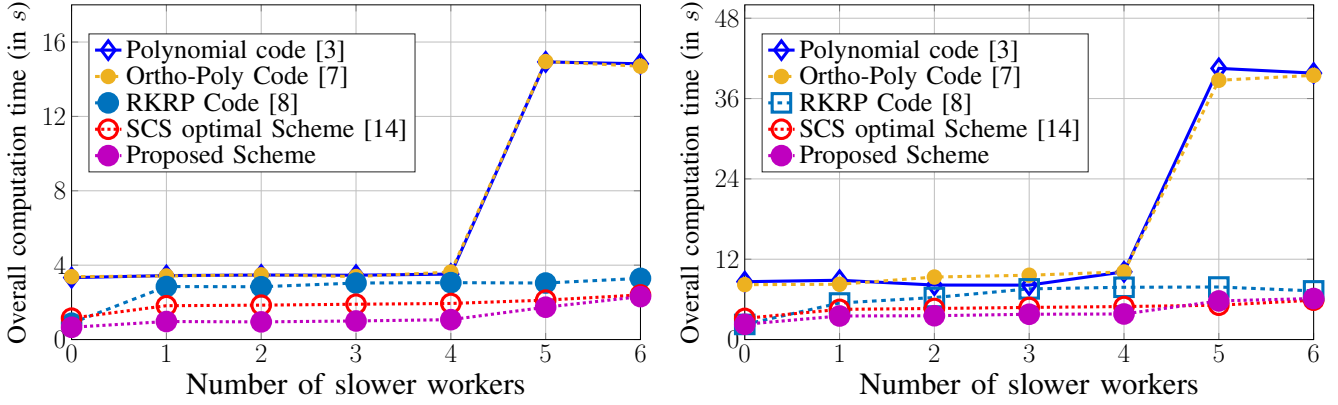

\subsection{Motivating Example}
\begin{example}
\label{ex:optsm}
\begin{figure}[t]
\centering
\resizebox{0.99\linewidth}{!}{
\begin{tikzpicture}[auto, thick, node distance=2cm, >=triangle 45]
\draw
    node [sum1, minimum size = 1.2cm, fill=green!30] (blk1) {$W_0$}
    node [sum1, minimum size = 1.2cm, fill=blue!30,right = 1.2 cm of blk1] (blk2) {$W_1$}
    node [sum1, minimum size = 1.2cm, fill=green!30,right = 1.2 cm of blk2] (blk3) {$W_2$}
    node [sum1, minimum size = 1.2cm, fill=blue!30,right = 1.2 cm of blk3] (blk4) {$W_3$}
    node [sum1, minimum size = 1.2cm, fill=green!30,right = 1.2 cm of blk4] (blk5) {$W_4$}

    node [block, minimum width = 2.2cm, fill=green!30,below = 0.4 cm of blk1] (blk11) {$\bfA_0$}
    node [block, minimum width = 2.2cm, fill=green!30,below = 0.0005 cm of blk11] (blk12) {$\bfA_1$}
    node [block, minimum width = 2.2cm, fill=green!30,below = 0.0005 cm of blk12] (blk13) {$\bfA_2$}
    node [block, minimum width = 2.2cm, fill=green!30,below = 0.0005 cm of blk13] (blk14) {$\bfA_3$}
    node [block,  minimum width = 2.2cm, fill=green!30,below = 0.0005 cm of blk14] (blk19) {$c_0 \bfA_4 + c_1 \bfA_9$}
    node [block,  minimum width = 2.2cm, fill=orange!50,below = 0.2 cm of blk19] (blk15) {$r_0 \bfB_0 + r_1 \bfB_1$}


    node [block, minimum width = 2.2cm, fill=blue!30,below = 0.4 cm of blk2] (blk21) {$\bfA_2$}
    node [block, minimum width = 2.2cm, fill=blue!30,below = 0.0005 cm of blk21] (blk22) {$\bfA_3$}
    node [block, minimum width = 2.2cm, fill=blue!30,below = 0.0005 cm of blk22] (blk23) {$\bfA_4$}
    node [block, minimum width = 2.2cm, fill=blue!30,below = 0.0005 cm of blk23] (blk24) {$\bfA_5$}
    node [block, minimum width = 2.2cm,  fill=blue!30,below = 0.0005 cm of blk24] (blk29) {$c_2 \bfA_6 + c_3 \bfA_1$}
    node [block,  minimum width = 1.5cm, fill=mycolor2!30,below = 0.2 cm of blk29] (blk25) {$r_2 \bfB_0 + r_3 \bfB_1$}

    node [block, minimum width = 2.2cm, fill=green!30,below = 0.4 cm of blk3] (blk31) {$\bfA_4$}
    node [block, minimum width = 2.2cm, fill=green!30,below = 0.0005 cm of blk31] (blk32) {$\bfA_5$}
    node [block, minimum width = 2.2cm, fill=green!30,below = 0.0005 cm of blk32] (blk33) {$\bfA_6$}
    node [block, minimum width = 2.2cm, fill=green!30,below = 0.0005 cm of blk33] (blk34) {$\bfA_7$}
    node [block, minimum width = 2.2cm,minimum width = 2.2cm, fill=green!30,below = 0.0005 cm of blk34] (blk39) {$c_4 \bfA_8 + c_5 \bfA_3$}
    node [block,  minimum width = 1.5cm,fill=orange!50,below = 0.2 cm of blk39] (blk35) {$r_4 \bfB_0 + r_5 \bfB_1$}

    node [block, minimum width = 2.2cm, fill=blue!30,below = 0.4 cm of blk4] (blk41) {$\bfA_6$}
    node [block, minimum width = 2.2cm, fill=blue!30,below = 0.0005 cm of blk41] (blk42) {$\bfA_7$}
    node [block, minimum width = 2.2cm, fill=blue!30,below = 0.0005 cm of blk42] (blk43) {$\bfA_8$}
    node [block, minimum width = 2.2cm, fill=blue!30,below = 0.0005 cm of blk43] (blk44) {$\bfA_9$}
    node [block, minimum width = 2.2cm, fill=blue!30,below = 0.0005 cm of blk44] (blk49) {$c_6 \bfA_0 + c_7 \bfA_5$}
    node [block,  minimum width = 1.5cm, fill=mycolor2!30,below = 0.2 cm of blk49] (blk45) {$r_6 \bfB_0 + r_7 \bfB_1$}

    node [block, minimum width = 2.2cm, fill=green!30,below = 0.4 cm of blk5] (blk51) {$\bfA_8$}
    node [block, minimum width = 2.2cm, fill=green!30,below = 0.0005 cm of blk51] (blk52) {$\bfA_9$}
    node [block, minimum width = 2.2cm, fill=green!30,below = 0.0005 cm of blk52] (blk53) {$\bfA_0$}
    node [block, minimum width = 2.2cm, fill=green!30,below = 0.0005 cm of blk53] (blk54) {$\bfA_1$}
    node [block, minimum width = 2.2cm, fill=green!30,below = 0.0005 cm of blk54] (blk59) {$c_8 \bfA_2 + c_9 \bfA_7$}  
    node [block,  minimum width = 1.5cm,fill=orange!30,below = 0.2 cm of blk59] (blk55) {$r_8 \bfB_0 + r_9 \bfB_1$}

    ;

\draw[->](blk1) -- node{} (blk11);
\draw[->](blk2) -- node{} (blk21);
\draw[->](blk3) -- node{} (blk31);
\draw[->](blk4) -- node{} (blk41);
\draw[->](blk5) -- node{} (blk51);

\end{tikzpicture}
}
\caption{\small Matrices $\bfA$ and $\bfB$ are partitioned into {\it ten} and {\it two} block-columns, respectively. Each worker is assigned four uncoded-coded and one coded-coded block-products. The coefficients $r_i$'s and $c_i$'s are chosen i.i.d. at random from a continuous distribution.}

\label{modopt_matmat}
\end{figure}
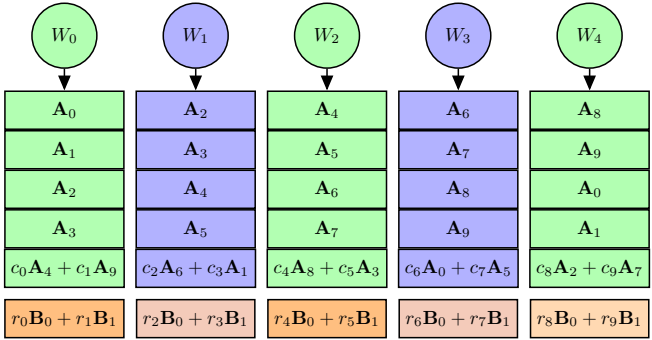 
Consider distributed matrix-matrix multiplication with $n = 5$ workers each of which can store $\gamma_A = \gamma_B = \frac{1}{2}$ portions of matrices $\bfA$ and $\bfB$, respectively. The job assignment according to our proposed approach is shown in Fig. \ref{modopt_matmat}, where matrices $\bfA$ and $\bfB$ are partitioned into $\Delta_A = 10$ and $\Delta_B = 2$ submatrices, respectively. In this case, each worker is responsible for computing all its corresponding pairwise block-products in a natural top-to-bottom order, e.g., worker $W_0$ computes $\bfA_0^T (r_0\bfB_0+r_1\bfB_1)$ followed by $\bfA_1^T (r_0\bfB_0+r_1\bfB_1), \dots, (c_0 \bfA_4 + c_1 \bfA_9)^T (r_0\bfB_0+r_1\bfB_1)$. Thus, each worker computes {\it four} uncoded-coded and {\it one} coded-coded block products. It can be verified that for these parameters, both the polynomial code approach and the approach in Fig. \ref{modopt_matmat} are resilient to {\it one} straggler.
\end{example}
We emphasize that the approach in Fig. \ref{modopt_matmat} has some notable advantages over the polynomial code approach. Suppose that matrices $\bfA$ and $\bfB$ are sparse. In the polynomial code approach, the encoded submatrices are approximately twice as dense as the original matrices. Thus, computing the product of the corresponding encoded matrices will take about twice as much time as compared to uncoded-coded block products in Fig. \ref{modopt_matmat}. 

Moreover, the latter approach can utilize the partial calculations done by the slower workers.  The polynomial code approach needs at least {\it four} workers to fully finish their respective assigned jobs. On the other hand, in Fig. \ref{modopt_matmat}, the final result can be recovered as soon as {\it any} $Q = 23$ block products are computed over all the workers according to the assigned computation order, where we recall that the central node needs to recover {\it twenty} submatrix products of the form $\bfA_i^T \bfB_j, 0 \leq i \leq 9, 0 \leq j \leq 1$. Thus, our approach allows us to leverage partial computations in scenarios where workers have differing speeds. It should be noted that the approach in \cite{das2020coded} provides $Q = 21$ for the same system, however, the coding for matrix $\bfA$ is denser in case of \cite{das2020coded}, which can lead to higher worker computation time.

\section{Matrix-matrix Multiplication Scheme}
\label{sec:matmat}
We now describe our proposed matrix-matrix multiplication scheme, beginning with an illustrative example that describes the encoding scheme and outlines the decoding scheme in case of worker node failures.



\subsection{An illustrative example}
\label{sec:illus_eg}
\input{new_xmpl_12}
Consider the example in Fig. \ref{new_matmat} where the system has $12$ worker nodes each of which can store $1/3$-rd fraction of $\bfA$ and $1/3$-rd fraction of $\bfB$. First we partition $\bfA$ into $12$ block-columns, and assign {\it three} uncoded and {\it one} coded submatrix of $\bfA$ to each of the worker nodes. The coded submatrices are always linear combinations of {\it three} uncoded submatrices. Similarly, we partition $\bfB$ into {\it three} block-columns, and assign only {\it one} coded $\bfB$ submatrix to each of the worker nodes; these are always linear combinations of {\it two} uncoded $\bfB$ submatrices.

There is a particular top-to-bottom order in which the tasks are executed within each worker node, e.g., in $W_0$ it is $\bfA_0^T (r_0\bfB_0 + r_1 \bfB_1), \bfA_1^T (r_0\bfB_0 + r_1 \bfB_1), \bfA_2^T (r_0\bfB_0 + r_1 \bfB_1)$ and finally $(c_{00} \bfA_3 + c_{01} \bfA_7+ c_{02}\bfA_{11})^T (r_0\bfB_0 + r_1 \bfB_1)$, where $c_{00}, c_{01}$ and $c_{02}$ represent the random coefficients for the encoded $\bfA$ submatrix in worker $W_0$.

To better understand the structure of the encoding scheme, consider the following class decomposition of $\bfA$ submatrices: $\calC_0 = \{\bfA_0, \bfA_4, \bfA_8\}, \calC_1 = \{\bfA_1, \bfA_5, \bfA_9\}, \calC_2 = \{\bfA_2, \bfA_6, \bfA_{10}\}, \calC_3 = \{\bfA_3, \bfA_7, \bfA_{11}\}$. The assignment of encoded $\bfA$ submatrices to each worker node are performed in a ``cyclic'' fashion, whereby a representative from each class is chosen for an uncoded assignment and a random linear combination of all members of the class is chosen for the coded assignment. For, instance in $W_0$ the three uncoded assignments are from (top-to-bottom) from classes $\calC_0, \calC_1$ and $\calC_2$ whereas the coded assignment is a random linear combination of the members of $\calC_3$. The sequence is shifted cyclically in $W_1$ and continues in a similar manner. The encoding of the $\bfB$ submatrices also follows a cyclic-pattern whereby $W_0$ contains a random linear combination of $\bfB_0$ and $\bfB_1$, $W_1$ contains a random linear combination of $\bfB_1$ and $\bfB_2$ and so on. 

It can be verified that each class (for example, $\calC_0$) appears in each worker node either as an uncoded assignment or as part of a coded assignment. 

It turns out that this scheme is resilient to the failure of any $s=3$ workers. While a general proof requires more ideas, we illustrate the recovery process by means of an example. Suppose $W_0$, $W_{10}$ and $W_{11}$ are failed so that all uncoded assignments of $\bfA_0$ are lost. Note that $\bfA_0 \in \calC_0$ and since we only encode the $\bfA$ submatrices within a certain class, it suffices for us to examine those submatrix products where class $\calC_0$ participates. From Fig. \ref{new_matmat} we can see that $\bfA_4 \in \calC_0$ and $\bfA_8 \in \calC_0$ participate in an uncoded manner in $W_2, W_3, W_4$ and $W_6, W_7, W_8$ respectively. The coded assigments corresponding to $\calC_0$ appear in $W_1, W_5$ and $W_9$.


Consider the following encoding matrices (with permuted columns for ease of viewing) that represent the coding coefficients for the submatrices of $\bfA$ and $\bfB$ respectively. The columns in left-to-right order correspond to worker nodes specified in the expression for $\bfW$ in (\ref{eq:worker_labels}) below. 

\begin{align}
\label{eq:GA_GB}
\centering
\bfG_A &= \begin{bmatrix}
\bovermat{Uncoded $\bfA_0$}{$1$ & $1$ & $1$} & & \bovermat{Uncoded $\bfA_4$}{$0$ & $0$ & $0$} & & \bovermat{Uncoded $\bfA_8$}{$0$ & $0$ & $0$} & & \bovermat{Coded $\calC_0$}{ * & * & *} \\
0 & 0 & 0 & & 1 & 1 & 1 & & 0 & 0 & 0 & & * & * & *\\
0 & 0 & 0 & & 0 & 0 & 0 & & 1 & 1 & 1 & & * & * & *\\
\end{bmatrix} ,\\ 
\bfG_B &= \begin{bmatrix}
* & 0 & * & & * & * & 0 & & * & 0 & * & & 0 & * & * \\
* & * & 0 & & 0 & * & * & & * & * & 0 & & * & 0 & * \\
0 & * & * & & * & 0 & * & & 0 & * & * & & * & * & 0 
\end{bmatrix} ,\\
\bfW &= \begin{bmatrix}
0 & 10 & 11 & $\,$ 2 & 3 & 4 & & 6 & 7 & 8 & & 1 & 5 & 9 $\,$ 
\end{bmatrix}. \label{eq:worker_labels}
\end{align}The asterisks in the above equations represent i.i.d. random choices from a continuous distribution.


Suppose for instance that $W_0, W_{10}$ and $W_{11}$ are stragglers. In this case one can observe that from $W_2, W_3$ and $W_4$ we can recover the products $\bfA_4^T \bfB_j$ for $j = 0,1,2$. This is because the corresponding system of equations for these three unknowns looks like
\begin{align}
    \begin{bmatrix}
 * & * & 0 \\
 0 & * & * \\
 * & 0 & * 
\end{bmatrix} ,
\end{align}
where the asterisks represent the chosen i.i.d. random coefficients. It can be seen that this matrix will be full rank with probability-1. In a similar manner we can argue that from $W_6, W_7$ and $W_8$ we can recover the products $\bfA_8^T \bfB_j, j = 0,1,2$. Following this, workers $W_1, W_5$ and $W_9$ allow the recovery of $\bfA_0^T \bfB_j, j = 0,1,2$. A more involved argument made in Section \ref{theorems} shows that in fact under any pattern of three stragglers the product $\bfA^T \bfB$ can be recovered.


The $Q/\Delta$ analysis is a bit more subtle. We divide the workers into three groups: $\calG_0 = \{ W_0, W_1, W_2, W_3 \}$, $\calG_1 = \{ W_4, W_5, W_6, W_7 \}$ and $\calG_3 = \{ W_8, W_9, W_{10}, W_{11} \}$. Within each group each class appears at all possible locations, e.g. within $\calG_0$, $\calC_0$ appears at location-0 in $W_0$, location-3 in $W_1$, location-2 in $W_2$ and location-1 in $W_3$ and so on. 
The appearances of $\calC_m$ in all different locations within a group allows us to leverage the properties of the cyclic assignment as done previously in \cite{c3les, das2020coded} and arrive at corresponding upper and lower bounds for $Q/\Delta$. We note here that the bounds match for $x=0$. As this is much more involved, we defer the argument to Section \ref{theorems}.



\subsection{Overview of Alg. \ref{Alg:New_Optimal_Matmat} }
\begin{algorithm}[t]
	\caption{Proposed scheme for distributed matrix-matrix multiplication}
	\label{Alg:New_Optimal_Matmat}
   \SetKwInOut{Input}{Input}
   \SetKwInOut{Output}{Output}
   \Input{Matrices $\bfA$ and $\bfB$, $n$-number of worker nodes, $s$-number of stragglers, storage fraction $\gamma_A = \frac{1}{k_A}$ and $\gamma_B = \frac{1}{k_B}$; $s \leq s_m = n - k_A k_B$.}
   Set $x = s_m - s$ and $y = \floor{\frac{k_A x}{s_m}}$\;
   Set $\Delta_A = \textrm{LCM}(n, k_A)$ and $\Delta_B = k_B$ and  Partition $\bfA$ and $\bfB$ into $\Delta_A$ and $\Delta_B$ block-columns, respectively\;
   Set $\Delta = \Delta_A \Delta_B$, $p =\frac{\Delta}{n}$ and $\ell = \frac{\Delta_A}{k_A}$\;
   Number of coded submatrices of $\bfA$ in each worker node, $\ell_c = \ell-p$\;
   Set $\omega = 1 + \ceil{\frac{s_m}{k_B}}$ and $\zeta = 1 + k_B - \Bigl\lceil{\frac{k_B}{\omega}\Bigr\rceil}$\; 
   Define $\calC_i = \left\lbrace \bfA_i, \, \bfA_{\ell+i}, \, \dots, \, \bfA_{(k_A-1)\ell+i} \right\rbrace$, and $\lambda_i = 0$, for $i = 0, 1, \dots, \ell-1$\;

   \For{$i\gets 0$ \KwTo $n-1$}{
   $u \gets i \times \frac{\Delta_A}{n}$\; 
   Define $T = \left\lbrace u, u+1, \dots, u + p - 1 \right\rbrace$ (modulo $\Delta_A$)\;
   Assign all $\bfA_{m}$'s sequentially from top to bottom to worker node $i$, where $m \in T$\;
   
   \For{$j\gets 0$ \KwTo $\ell_c - 1$}{
   $v \gets u + p + j$ (mod $\ell$)\;
   Denote $\bfY \in \calC_v$ as the set of the element submatrices at locations (modulo $k_A$) $\lambda_v, \lambda_v+1, \lambda_v+2, \dots, \lambda_v+k_A-y-1$ of $\calC_v$\; 
   Assign a random linear combination of $\bfA_{q}$'s where $\bfA_q \in \bfY$\;
   $\lambda_v \gets \lambda_v + k_A - y$ (modulo $k_A$)\;
   }

   Define $V = \left\lbrace i, i+1, \dots, i + \zeta - 1 \right\rbrace$ (modulo $\Delta_B$)\;
   Assign a random linear combination of $\bfB_{q}$'s where $\bfB_q \in \bfV$\;

   }
   \Output{$\langle n, \gamma_A, \gamma_B \rangle$-scheme for distributed matrix-matrix multiplication.}
\end{algorithm}


We now discuss the scheme specified formally in Algorithm \ref{Alg:New_Optimal_Matmat}. The symbols and notation introduced in Algorithm \ref{Alg:New_Optimal_Matmat} are summarized in Table \ref{notation}.

\noindent {\it Weight of the linear combination of $\bfA$ and $\bfB$ submatrices:} Note that $s_m = n - k_A k_B$ is the maximum number of stragglers that the scheme can be resilient to, whereas we want resilience to $s \leq s_m$ stragglers. Line 1 in Alg. \ref{Alg:New_Optimal_Matmat} sets the parameter $x= s_m - s$. Thus, $x$ measures the relaxation of the straggler resilience that we are able to tolerate. This allows us to reduce the weight of the linear combination of the $\bfA$ submatrices. In particular, let $y = \floor{\frac{k_A x}{s_m}}$. Then, our algorithm combines at most $k_A - y$ submatrices of $\bfA$. 

The encoded submatrices of $\bfB$ are obtained by combining $\zeta$ submatrices of $\{\bfB_0, \bfB_1, \dots, \bfB_{\Delta_B - 1}\}$. Line 5 specifies the assignment of $\zeta$; it can be observed that $\zeta \leq \Delta_B = k_B$.

\noindent {\it Assignment of encoded submatrices of $\bfA$:} We further divide the set $\{\bfA_0, \bfA_1, \dots, \bfA_{\Delta_A -1} \}$ into $\ell$ disjoint classes $\calC_0, \calC_1, \dots, \calC_{\ell-1}$, i.e.,
\begin{align}
\label{eq:equivalence}
    \calC_m = \left\lbrace \bfA_m, \, \bfA_{\ell+m}, \, \bfA_{2\ell+m}, \, \dots, \, \bfA_{(k_A-1)\ell+m} \right\rbrace.
\end{align} This implies that $|\calC_m| = k_A$, for $m= 0, 1, \dots, \ell-1$, and submatrix $\bfA_i$ belongs to $\calC_{i \, (\textrm{mod} \, \ell)}$. 

The worker nodes are assigned submatrices from each class $\calC_m, 0 \leq m \leq \ell-1$ in a block-cyclic fashion; the block shift is specified by $\Delta_A/n$ (line 8). In each worker node, the first $p = \Delta/n$ assignments are uncoded, i.e., they correspond to a specific element of the corresponding class. The remaining $\ell_c = \ell - p$ assignments are coded. Each coded assignment corresponds to random linear combination of an appropriate $(k_A - y)$-sized subset of the corresponding class. This is discussed in line 8 -- 16 in Alg. \ref{Alg:New_Optimal_Matmat}. 

As each location of every worker node is populated by a submatrix from a class $\calC_m$ where $0 \leq m \leq \ell-1$, we will occasionally say that the class $\calC_m$ appears at a certain location (between $0$ to $\ell-1$) at a certain worker node. To ensure that each submatrix of $\calC_m$ participates in ``almost'' the same number of coded assignments, we use a counter $\lambda_i$ to keep track of the linear combination that will be formed from the corresponding class $\calC_i,  0 \leq i \leq \ell-1$ (lines 6, 13 -- 15 in Alg. \ref{Alg:New_Optimal_Matmat}).

In Section \ref{sec:illus_eg}, we discussed an example where $x=0$ so that $y=0$. Therefore, the encoded $\bfA$ matrices combine all the three submatrices within the respective classes. There are $p=\Delta/n=3$ uncoded $\bfA$ submatrices and one coded $\bfA$ submatrix in each worker node.


\noindent {\it Assignment of encoded submatrices of $\bfB$:} For worker $W_i$, consider the set $V = \{i, i+1, \dots, i + \zeta -1\} \, (\textrm{mod} \; \Delta_B)$. A random linear combination of $\bfB_k$ for $k \in V$ is assigned to worker $W_i$. We note here that $\zeta \leq k_B$ and can in fact be as small as $\lceil k_B/2 \rceil$ depending upon the values of $k_B$ and $s_m$ (see line 5 of Alg. \ref{Alg:New_Optimal_Matmat}).

For instance, in Fig. \ref{new_matmat}, $s_m = 3$ and $k_B = 3$, so that $\omega = 2$ which implies that $\zeta = 2$. Thus, for instance for worker $W_8$, the set $V = \{8,9\}\, (\textrm{mod} \; \Delta_B) = \{2,0\}$ (where $\Delta_B = k_B$) and it is assigned a random linear combination of $\bfB_2$ and $\bfB_0$. Moreover, the patterns repeats periodically.

\noindent {\it Order of jobs:} Note that each worker node is only assigned one encoded $\bfB$ submatrix. Each worker node computes the product of its assigned $\bfA$ submatrices with the corresponding encoded $\bfB$ submatrix in the top to bottom order.


In the following subsections we point out certain ``structural'' properties of our scheme. In the presence of $s$ stragglers, suppose that there is a submatrix $\bfA_i^T \bfB_j$ where $\bfA_i \in \calC_m$ that we cannot decode. Our scheme is such that we can just focus on the equations where $\calC_m$ participates. This provides a manageable subset of equations where we can focus our attention. Different properties of the scheme (Lemma \ref{lem:new_opt_matmat} and Claim \ref{claim:appearB}) allow us to assert that the overall system of equations seen by submatrices $\bfA_i^T \bfB_j$ where $\bfA_i \in \calC_m$ and $j = 0, 1, \dots, \Delta_B-1$ is full-rank even in the presence of $s$ stragglers (Theorem \ref{thm:matmatstr}).

\subsection{Coding for Matrix $\bfA$}
\label{sec:codingA}

Let $\bfU_i$ denote the subset of worker nodes where $\bfA_i$ appears in an uncoded block, for $i = 0, 1, \dots, \Delta_A-1$. Likewise, $\bfV_i$ denotes the subset of worker nodes where $\bfA_i$ appears in a coded block. Our first claim states that the number of coded appearances of any two submatrices in a class can differ by at most $one$. The detailed proof is given in Appendix \ref{App:proofclaim1}.

\begin{claim}
\label{claim:diffV} 
If the jobs are assigned to the workers according to Alg. \ref{Alg:New_Optimal_Matmat}, for any $\bfA_i, \bfA_j \in \calC_m$,
\begin{align*}
\Bigl||\bfV_i| - |\bfV_j|\Bigr| \leq 1 .
\end{align*} 
\end{claim}

We now present a lemma which outlines the key properties of the structure of encoding submatrices of $\bfA$. It includes the details on how a given submatrix $\bfA_i$ and the different classes appear at different locations over all the worker nodes. The proof of this lemma is detailed in Appendix \ref{App:prooflemma1}.

\begin{lemma}
\label{lem:new_opt_matmat}
Assume that the jobs are assigned to the workers according to Alg. \ref{Alg:New_Optimal_Matmat}, and consider any submatrix $\bfA_i$, for $i = 0, 1, 2, \dots, \Delta_A - 1$. Then (i) $|\bfU_i| = k_B$, (ii) $|\bfV_i| \geq s$ and $\bfU_i \cap \bfV_i = \emptyset$, and (iii) a given class $\calC_m$, where $0 \leq m \leq \ell - 1$, appears at all different locations $0, 1, \dots, \ell-1$ within the worker nodes of any worker group $\calG_\lambda$, where $0 \leq \lambda \leq c-1$.
\end{lemma}

\begin{example}
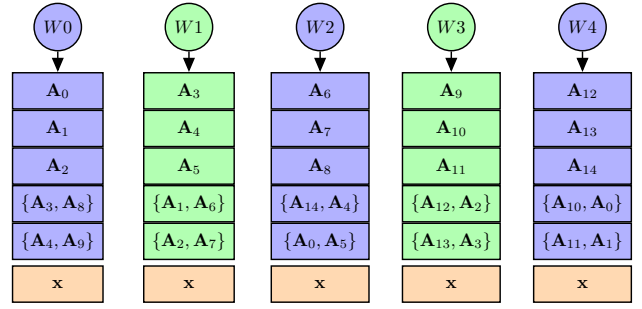
\begin{figure}[t]
\centering
\resizebox{0.95\linewidth}{!}{
\begin{tikzpicture}[auto, thick, node distance=2cm, >=triangle 45]

\draw

    node [sum, fill=blue!30] (blk1) {$W0$}
    node [sum, fill=green!30,right = 1.6 cm of blk1] (blk2) {$W1$}
    node [sum, fill=blue!30,right = 1.6 cm of blk2] (blk3) {$W2$}
    node [sum, fill=green!30,right = 1.6 cm of blk3] (blk4) {$W3$}
    node [sum, fill=blue!30,right = 1.6 cm of blk4] (blk5) {$W4$}

    node [block, fill=blue!30, minimum width = 5 em, below = 0.4 cm of blk1] (blk11) {$\bfA_0$}
    node [block, fill=blue!30, minimum width = 5 em, below = 0.0005 cm of blk11] (blk12) {$\bfA_1$}
    node [block, fill=blue!30, minimum width = 5 em, below = 0.0005 cm of blk12] (blk13) {$\bfA_2$}
    node [block, fill=blue!30, minimum width = 5em, below = 0.0005 cm of blk13] (blk14) {$\{\bfA_3, \bfA_8\}$}
    node [block, fill=blue!30, minimum width = 5em, below = 0.0005 cm of blk14] (blk16) {$\{\bfA_4, \bfA_9\}$}
    node [block, fill=orange!30,minimum width = 5 em,below =  0.1 cm of blk16] (blk17) {$\bfx$}

    node [block, fill=green!30,minimum width = 5 em,below = 0.4 cm of blk2] (blk21) {$\bfA_3$}
    node [block, fill=green!30,minimum width = 5 em,below = 0.0005 cm of blk21] (blk22) {$\bfA_4$}
    node [block, fill=green!30,minimum width = 5 em,below = 0.0005 cm of blk22] (blk23) {$\bfA_5$}
    node [block, fill=green!30,minimum width = 5em,below = 0.0005 cm of blk23] (blk24) {$\{\bfA_1, \bfA_6\}$}
	node [block, fill=green!30,minimum width = 5em,below = 0.0005 cm of blk24] (blk26) {$\{\bfA_2, \bfA_7\}$}
	node [block, fill=orange!30,minimum width = 5em,below =  0.1 cm of blk26] (blk27) {$\bfx$}
	
    node [block, fill=blue!30,minimum width = 5.4 em,below = 0.4 cm of blk3] (blk31) {$\bfA_6$}
    node [block, fill=blue!30,minimum width = 5.4 em,below = 0.0005 cm of blk31] (blk32) {$\bfA_7$}
    node [block, fill=blue!30,minimum width = 5.4 em,below = 0.0005 cm of blk32] (blk33) {$\bfA_8$}
    node [block, fill=blue!30,minimum width = 5.4 em,below = 0.0005 cm of blk33] (blk34) {$\{\bfA_{14}, \bfA_4\}$}
    node [block, fill=blue!30, below = 0.0005 cm of blk34, minimum width = 5.4 em] (blk36) {$\{\bfA_0, \bfA_5\}$}
	node [block, fill=orange!30,minimum width = 5.4 em, below =  0.1 cm of blk36] (blk37) {$\bfx$}
	
    node [block, fill=green!30,minimum width = 5.4 em, below = 0.4 cm of blk4] (blk41) {$\bfA_9$}
    node [block, fill=green!30,minimum width = 5.4 em,below = 0.0005 cm of blk41] (blk42) {$\bfA_{10}$}
    node [block, fill=green!30,minimum width = 5.4 em,below = 0.0005 cm of blk42] (blk43) {$\bfA_{11}$}
    node [block, fill=green!30,minimum width = 5.4 em,below = 0.0005 cm of blk43] (blk44) {$\{\bfA_{12}, \bfA_2\}$}
    node [block, fill=green!30,minimum width = 5.4 em, below = 0.0005 cm of blk44] (blk46) {$\{\bfA_{13}, \bfA_3\}$}
	node [block, fill=orange!30,minimum width = 5.4 em, below =  0.1 cm of blk46] (blk47) {$\bfx$}
	
	node [block, fill=blue!30,minimum width = 5.4 em, below = 0.4 cm of blk5] (blk51) {$\bfA_{12}$}
    node [block, fill=blue!30,minimum width = 5.4 em,below = 0.0005 cm of blk51] (blk52) {$\bfA_{13}$}
    node [block, fill=blue!30,minimum width = 5.4 em,below = 0.0005 cm of blk52] (blk53) {$\bfA_{14}$}
    node [block, fill=blue!30,minimum width = 5.4 em,below = 0.0005 cm of blk53] (blk54) {$\{\bfA_{10}, \bfA_0\}$}
    node [block, fill=blue!30,minimum width = 5.4 em, below = 0.0005 cm of blk54] (blk56) {$\{\bfA_{11}, \bfA_1\}$}
	node [block, fill=orange!30,minimum width = 5.4 em, below =  0.1 cm of blk56] (blk57) {$\bfx$}
;
\draw[->](blk1) -- node{} (blk11);
\draw[->](blk2) -- node{} (blk21);
\draw[->](blk3) -- node{} (blk31);
\draw[->](blk4) -- node{} (blk41);
\draw[->](blk5) -- node{} (blk51);

\end{tikzpicture}
}
\caption{\small Matrix-vector case with $n = 5$ workers and $s = 1$ stragglers, with $\gamma_A = \frac{1}{3}$. Here, $s_m = n - \frac{1}{\gamma_A} = 2$. In this figure, any $\{\bfA_i, \bfA_j\}$ means a linear combination of those submatrices, where the coefficients are chosen i.i.d. at random from a continuous distribution.}
\label{matvec}
\vspace{-0.1in}
\end{figure} 
To clarify the idea of the proof of Lemma \ref{lem:new_opt_matmat}(ii), We consider an example with distributed matrix-vector multiplication (which is equivalent to $k_B = 1$ in distributed matrix-matrix multiplication) in Fig. \ref{matvec} where we have a system with $n = 5$ workers, $\gamma_A = \frac{1}{3}$ and $s = 1$. We consider the class, $\calC_0 = \{ \bfA_0, \bfA_5, \bfA_{10} \}$. It can be verified from Fig. \ref{matvec} that 
\begin{align*}
    \mu_0 = \frac{|\bfV_0| + |\bfV_5| + |\bfV_{10}|}{3} = \frac{2 + 1 + 1}{3} = \frac{4}{3}.
\end{align*} So, $\floor{\mu_0} = 1$ and $\ceil{\mu_0} = 2$; which satisfies the inequality 
\begin{align*}
\floor{\mu_0} \leq |\bfV_0|, |\bfV_5|, |\bfV_{10}| \leq \ceil{\mu_0}.    
\end{align*} Thus, $|\bfV_i| \geq 1 = s$; for any $i = 0, 1, \dots, \Delta_A - 1$.
\end{example}

The following corollary states that the submatrices in $\calC_m$ are assigned to $k_A k_B$ distinct workers as uncoded blocks and to the remaining $s_m = n - k_A k_B$ workers as coded blocks. The proof appears in Appendix \ref{App:proofcorr2}.

\begin{corollary}
\label{cor:matmat}
If $\calC_m = \left\lbrace \bfA_m, \, \bfA_{\ell+m}, \, \dots, \, \bfA_{(k_A-1)\ell+m} \right\rbrace$, then
\begin{align*}
&(i) \;\; \abs*{ \left( \underset{i: \bfA_i \in \calC_m }{\cup} \bfU_i\right) } = k_A k_B \; \; \textrm{and} \;\; \abs*{ \left( \underset{i: \bfA_i \in \calC_m }{\cup} \bfV_i\right) } = s_m \, ; \\
(&ii)  \; \; \left( \underset{i: \bfA_i \in \calC_m }{\cup} \bfU_i\right)  \; \cap \;  { \left( \underset{i: \bfA_i \in \calC_m }{\cup} \bfV_i\right) } \; = \; \emptyset \, .
\end{align*} 
\end{corollary}

\subsection{Coding for Matrix $\bfB$}
\label{sec:codingB}
To discuss the coding for matrix $\bfB$, first we consider a $k_B \times n$ matrix, where each column has $\zeta \leq k_B$ non-zero entries which are chosen i.i.d. from a continuous distribution. Moreover, the indices of non-zero entries are consecutive and shifted in a cyclic fashion, reduced modulo $k_B$. For example, if we have a system with $n = 12$ workers with $k_A = 2$ and $k_B = 5$, then $\zeta = 3$ and the corresponding coding matrix for $\bfB$, denoted as $\bfR^B_{k_B,n}$, can be written as 
\begin{align}
\label{eq:Grb}
    \bfR^B_{k_B,n} = \begin{bmatrix}[c c c c c c c c c c c c]
     * & 0 & 0 & * & * & * & 0 & 0 & * & * & * & 0 \\
     * & * & 0 & 0 & * & * & * & 0 & 0 & * & * & * \\
     * & * & * & 0 & 0 & * & * & * & 0 & 0 & * & *  \\
     0 & * & * & * & 0 & 0 & * & * & * & 0 & 0 & * \\
     0 & 0 & * & * & * & 0 & 0 & * & * & * & 0 & 0 \\
\end{bmatrix} \; ;
\end{align} where $*$ indicates the non-zero entries. The entries at indices $i, i+1, \dots, i+\zeta - 1$ (reduced modulo $k_B$) are non-zero (chosen i.i.d. from a continuous distribution) within column $i$ of $\bfR^B_{k_B,n}$ and the other entries are set to zero. The non-zero coefficients are used to specify the random linear combination of the submatrices of $\bfB$ assigned to worker $W_i$.

\begin{definition}
\label{def:typeB}
A type $i$ submatrix, for $i = 0, 1, 2 \dots, k_B - 1$, is a random linear combination of the submatrices, $\bfB_i, \bfB_{i+1}, \dots, \bfB_{i + \zeta - 1}$ (indices reduced modulo $k_B$). Thus we can say that worker node $W_j$ is assigned a type $j \,$ (mod $\; k_B)$ submatrix (line 18 of Alg. \ref{Alg:New_Optimal_Matmat}).
\end{definition}


Consider the case of $x = 0$ and any $\bfA_i$, $i = 0, 1, 2, \dots, \Delta_A - 1$. From Lemma \ref{lem:new_opt_matmat}, we know $|\bfU_i| = k_B$ and $|\bfV_i| = s_m$ (since $x = 0, s = s_m$). Thus $\bfA_i$ appears at $\sigma = k_B + s_m$ worker nodes. We now investigate the types   ({\it cf.} Def. \ref{def:typeB}) of the coded submatrices of $\bfB$ in those $\sigma$ worker nodes. The following claim specifies the types of those $\bfB$ submatrices, and the proof is given in Appendix \ref{App:proofclaim2}.

\begin{claim}
\label{claim:appearB}
Consider the construction in Alg. \ref{Alg:New_Optimal_Matmat} with $x=0$ and let $k$ be the minimum index of the worker node where $\bfA_i$ appears (uncoded or coded) and consider the worker nodes in $\bfU_i \cup \bfV_i$. The assigned submatrices of $\bfB$ for those worker nodes are, respectively, from types $k, k+1, k+2, \dots, k + \sigma - 1$ (reduced modulo $k_B$), which are $\sigma$ consecutive types.
\end{claim}

\subsection{Straggler Resilience and bounds on $Q/\Delta$}
\label{theorems}
\begin{lemma}
\label{lem:graphtheory}
Consider any $\bfA_i$, $i = 0, 1, 2, \dots, \Delta_A - 1$, for the case of $x = 0$. Construct a $k_B \times \sigma$ matrix $\bfR_i$ where the columns of $\bfR_i$ correspond to the coefficients for coded $\bfB$ submatrices of the worker nodes in $\bfU_i \cup \bfV_i$ and $\sigma = k_B + s_m$. If $\zeta > k_B - \Bigl\lceil{\frac{k_B}{\omega}\Bigr\rceil}$, any $k_B \times k_B$ submatrix of $\bfR_i$ is full rank, where $\omega = 1 + \ceil{\frac{s_m}{k_B}}$. 
\end{lemma}

\begin{proof}
The proof is given in Appendix \ref{App:prooflemma2}.
\end{proof}

For any class $\calC_m$, the encoded submatrices of $\bfA$ within different worker nodes can be specified in terms of a $k_A \times n$ ``generator'' matrix, e.g., in Fig. \ref{modopt_matmat}, the column of the generator matrix for worker $W_0$ corresponding to $\calC_4 = \{ \bfA_4, \bfA_9\}$ will be $[c_0~~c_1]^T$. Similarly the encoded submatrices of $\bfB$ within different worker nodes can be specified in terms of a $k_B \times n$ ``generator'' matrix, e.g., in Fig. \ref{modopt_matmat}, the column of the generator matrix for worker $W_0$ will be $[r_0~~r_1]^T$. We use this formalism in the discussion below.

\begin{theorem}
\label{thm:matmatstr}
Assume that a system has $n$ worker nodes each of which can store $1/k_A$ and $1/k_B$ fraction of matrices $\bfA$ and $\bfB$, respectively. In each worker, according to Alg. \ref{Alg:New_Optimal_Matmat}, we assign some uncoded $\bfA$ submatrices and some coded $\bfA$ submatrices with weight $k_A - y$ , where $s_m = n - k_A k_B$ and $y = \floor{\frac{k_A x}{s_m}}$. We also assign a coded $\bfB$ submatrix to each worker which has a weight $\zeta$, as described in Lemma \ref{lem:graphtheory}. Then, this distributed matrix-matrix multiplication scheme will be resilient to $s = s_m - x$ stragglers.
\end{theorem}


\begin{proof}
We assume that there are $s = s_m - x$ stragglers, and we cannot recover an unknown $\bfA_i^T \bfB_j$ from the remaining $\tau = k_A k_B + x$ workers. Let $\bfA_i \in \calC_m$, for some $m = 0, 1, 2, \dots, \ell - 1$. From Lemma \ref{lem:new_opt_matmat} part (i), $|\bfU_m| = |\bfU_{\ell+m}| = \dots = |\bfU_{(k_A-1)\ell+m}| = k_B$. Thus, from Corollary \ref{cor:matmat}, without loss of generality, by permuting the columns appropriately, the $k_A \times n$ generator matrix for the corresponding submatrices of $\calC_m$ can be expressed as
\begin{align*}
    \bfG_A  =   \begin{bmatrix}[cccc|cc]
\; \mathbf{1}_{k_B}  & \mathbf{0} &\dots & \mathbf{0} &  \\
\;\mathbf{0} & \mathbf{1}_{k_B}  &\dots & \mathbf{0} & \\
\;\mathbf{0} & \mathbf{0}  &\dots & \mathbf{0} & \bfR^A_{k_A,s_m} \\
\;\threevdots & \threevdots & \dots & \threevdots  &  \\
\;\mathbf{0} & \mathbf{0}  & \dots & \mathbf{1}_{k_B} &
\end{bmatrix}, 
\end{align*} where $\mathbf{1}_{k_B}$ represents an all-one row-vector of length $k_B$. An example of this is shown in Section \ref{sec:illus_eg}. Here $\bfR^A_{k_A,s_m}$ is a random matrix of size $k_A \times s_m$ whose each column has $k_A - y$ non-zero entries (and each row has at least $s = s_m - x$ non-zero entries) where $y$ is defined in Line 1 in Alg. \ref{Alg:New_Optimal_Matmat}. The first $k_A k_B$ columns of $\bfG_A$ denote the uncoded submatrices of $\calC_m$ and the next $s_m = n - k_A k_B$ columns denote the coded submatrices. Similarly, the generator matrix for the corresponding coded submatrices of $\bfB$ is $\bfG_B  = \bfR^B_{k_B,n}$ (as mentioned in \eqref{eq:Grb}). Thus, the generator matrix for the unknowns of the form $\bfA^T_{\alpha} \bfB_{\beta}$ (where $\bfA_{\alpha} \in \calC_m$, $\beta = 0, 1, \dots, \Delta_B - 1$) is given by $\bfG = \bfG_{A} \odot \bfG_{B}$ ($\odot$ denotes the Khatri-Rao product \cite{zhang2017matrix} which corresponds to column-wise Kronecker product) which is of size $k_A k_B \times n$. The following lemma states a relevant rank property of $\bfG$, and the corresponding proof is given in Appendix \ref{sec:Gmmproof}.

\begin{lemma}
\label{lem:Gmmproof}
Any $k_A k_B \times \tau$ submatrix of $\bfG$ has a rank $k_A k_B$ with probability $1$, where $\tau = k_A k_B + x$.
\end{lemma}

The unknowns corresponding to $\calC_m$ can be represented in terms of the following Kronecker product as
\begin{align*}
    \bfy &=  \begin{bmatrix} \bfA_m^T   &  \bfA_{\ell+m}^T  \, \dots  \bfA_{(k_A-1)\ell+m}^T \end{bmatrix} \otimes \begin{bmatrix} \bfB_0 & \bfB_1 \, \dots  \bfB_{k_B- 1}\end{bmatrix} \\ 
    & = \begin{bmatrix}
     \bfA_m^T \bfB_0 & \bfA_m^T \bfB_1 & \dots & \dots & \bfA_{(k_A-1)\ell+m}^T \bfB_{k_B-1}
    \end{bmatrix},
\end{align*} thus there are $k_A k_B$ such unknowns of the form $\bfA_{\alpha}^T \bfB_{\beta}$, where $\bfA_{\alpha} \in \calC_m$. Note that $\bfA_i^T \bfB_j$ is also one of them which is assumed to be not decodable from the $\tau$ workers. But from Lemma \ref{lem:Gmmproof}, we can show that any $k_A k_B \times \tau$ submatrix of $\bfG$ has a rank $k_A k_B$ with probability $1$, which indicates that all $k_A k_B$ unknowns corresponding to $\calC_m$ can be recovered from any $\tau$ workers. This contradicts our assumption that $\bfA_i^T \bfB_j$ is not decodable. 
\end{proof}

Let us consider the example in Fig. \ref{modopt_matmat}, and apply the argument for unknown $\bfA_0^T \bfB_0$. Here we have $\ell = 5$, so $\bfA_0 \in \calC_0 = \{\bfA_0, \bfA_5 \}$. Thus to recover the unknowns of the form $\bfA^T_{\alpha} \bfB_{\beta}$ corresponding to $\calC_0$, we have the corresponding generator matrices as
\begin{align*}
    \bfG_A &= \begin{bmatrix} 1 & 1 & 0 & 0 & c_6 \\
    0 & 0 & 1 & 1 & c_7 \\
    \end{bmatrix},  \; \; \; \textrm{and} \\
    \bfG_B &= \begin{bmatrix} r_0 & r_8 & r_2 & r_4 & r_6 \\
    r_1 & r_9 & r_3 & r_5 & r_7 \\
    \end{bmatrix};
\end{align*} where the columns correspond to $W_0, W_4, W_1, W_2$ and $W_3$, respectively. Next we can have the generator matrix $\bfG = \bfG_{A} \odot \bfG_{B}$ having a size $4 \times 5$ whose any $4 \times 4$ square submatrix is full-rank. Thus we can recover the unknowns, $\bfA_0^T \bfB_0, \bfA_0^T \bfB_1, \bfA_5^T \bfB_0$ and $\bfA_5^T \bfB_1$ from the results returned by any four workers.

We now present the result of our work on utilizing the partial computations. It provides the calculation of the value of $Q$ for our scheme for different system parameters. Before stating the corresponding theorem, we state the following claim, which follows from \cite{das2020coded}. The proof is detailed in Appendix \ref{app:proofclaim3}.
\begin{claim}
\label{claim:Qprev}
Assume that the jobs are assigned to the workers according to Alg. \ref{Alg:New_Optimal_Matmat}, and consider any class $\calC_m$, for $m = 0, 1, \dots, \ell - 1$. The maximum number of submatrix-products that can be acquired from the job assignments where $\calC_m$ appears exactly $\kappa - 1$ times is
\begin{align*}
 \eta = \frac{n (\ell - 1 )}{2} + c \sum\limits_{i=0}^{c_1 - 1} (\ell - i) + c_2 (\ell - c_1) \, ;
\end{align*}  where $c = \frac{n}{\ell}$, $c_1 = \floor{\frac{\kappa - 1}{c}}$, and $c_2 = \kappa - 1 - c c_1$.   
\end{claim}

\begin{theorem}
\label{thm:matmatq}
Alg. \ref{Alg:New_Optimal_Matmat} proposes a distributed matrix-matrix multiplication scheme which provides $Q$ such that $Q_{lb} \leq Q \leq Q_{ub}$. Here the bounds are given by
\begin{align*}
Q_{ub} &= \frac{n (\ell - 1 )}{2} + c \sum\limits_{i=0}^{c^x_1 - 1} (\ell - i) + c^x_2 (\ell - c^x_1) + 1 \, ; \;\textrm{and}\\
Q_{lb} &= \frac{n (\ell - 1 )}{2} + c \sum\limits_{i=0}^{c^0_1 - 1} (\ell - i) + c^0_2 (\ell - c^0_1) + \Bigl\lceil{ \frac{s_m y}{k_A}\Bigr\rceil}  + 1 \, ;
\end{align*}  where $c = \frac{n}{\ell}$, $c^x_1 = \Bigl\lfloor{\frac{k_A k_B + x - 1}{c} \Bigr\rfloor}$, $c^x_2 = k_A k_B + x - 1 - c c^x_1$ and $y = \Bigl\lfloor{ \frac{k_A x}{s_m}\Bigr\rfloor}$.
\end{theorem}
\begin{proof}
The proof of this theorem is given in Appendix \ref{app:proofQ}.
\end{proof}

When $x = 0$, then $\tau = k_A k_B$, $c_1^x = c_1^0$ and $c_2^x = c_2^0$, hence $Q_{lb} = Q_{ub} = Q$. Let us consider a special case, where $n$ and $k_A$ are co-prime. In that case, $\Delta_A = n \times k_A$, so $\ell = \Delta_A/k_A = n$ and $c = 1$. Moreover, $c_1^0 = k_A k_B - 1$ and $c_2^0 = 0$. Thus, we have 
\begin{align*}
    Q & = \frac{n (n - 1)}{2} + \sum\limits_{i=0}^{k_A k_B - 1} (n - i) + 1  \\
    & =  \frac{n (n - 1)}{2} + n k_A k_B - \frac{k_A k_B (k_A k_B - 1)}{2} + 1 \\
    & = n k_A k_B + \frac{(n - k_A k_B) (n + k_A k_B - 1)}{2} + 1 \\
    & \approx \Delta + \frac{s_m \times 2 k_A k_B}{2} + 1
\end{align*} when $s_m = n - k_A k_B$ is very small. Thus, we have $\frac{Q}{\Delta} \approx 1 + \frac{s_m}{n}$. If $s_m \ll n$, then $Q/\Delta$  is very close to $1$, which indicates that, in this special case, the proposed scheme can efficiently utilize the partial computations done by the slower workers.

It should be noted that the trivial lower bound of $Q$ is $\Delta$. This can be achieved directly by assigning multiple evaluations in many of the dense coded approaches \cite{yu2017polynomial, 8849468, 8919859}. But, the issue here is sparsity. 
The weights of the encoded $\bfA$ and $\bfB$ matrices in those approaches are $k_A$ and $k_B$, respectively, which can destroy the inherent sparsity of the matrices, so the worker computation time can go up significantly.

In our proposed approach, the value of $Q$ is slightly more than $\Delta$. However, we assign many uncoded $\bfA$ submatrices and we reduce the weight of $\bfB$ submatrices, which help to preserve the sparsity of $\bfA$ and $\bfB$ and can reduce the worker computation time. Thus we can gain in overall computation time as shown in Fig. \ref{over_comp_sparse} even though we lose a small amount in the $Q/\Delta$ metric. 

\begin{example}
We consider an example with $n = 8$ and $\gamma_A = \frac{1}{3}, \gamma_B = \frac{1}{2}$. So, we partition $\bfA$ into $\Delta_A = \textrm{LCM}(n, k_A) = 24$ submatrices and $\bfB$ into $\Delta_B = k_B = 2$ submatrices. The properties of this scheme are discussed in Table \ref{table:n8k6} for different values of $x$. For $x = 0$, the recovery threshold is $6$, and $Q_{lb} = Q_{ub} = Q = 59$. Moreover, for $x = 1$, the recovery threshold is $7$ and $Q_{lb} \leq Q \leq Q_{ub}$ where $Q_{ub} - Q_{lb} =  2$. 

\begin{table}[t]
\caption{{\small Comparison of properties of the system with $n =8$ and $\gamma_A = \frac{1}{3}$ and $\gamma_B = \frac{1}{2}$ for different values of $x$}}
\label{table:n8k6}
\begin{center}
\begin{small}
\begin{sc}
\begin{tabular}{c c c c c c}
\hline
\toprule
$\; \; x \; \; $ &  $\; \; y \; \;$ & $\; \; \tau \; \;$ & $\; \; Q_{lb}\; \;$ & $\; \;Q_{ub}\; \;$ & $\; \; Q \; \;$ \\
 \midrule
$\; \; 0 \; \; $ &  $\; \; 0 \; \;$ & $\; \; 6 \; \;$ & $\; \; 59\; \;$ & $\; \;59\; \;$ & $\; \; 59 \; \;$ \\
$\; \; 1 \; \; $ &  $\; \; 1 \; \;$ & $\; \; 7 \; \;$ & $\; \; 60\; \;$ & $\; \;62\; \;$ & $\; \; 61 \; \;$ \\
\bottomrule
\end{tabular}
\end{sc}
\end{small}
\end{center}
\end{table}%
\end{example}

\subsection{Dealing with sparse input matrices} 
\label{sec:sparsity}

We now discuss the performance of our scheme when the input matrices are sparse. In our algorithm, among the $\ell$ submatrices of $\bfA$, we assign $\ell_u = \frac{\Delta_A \Delta_B}{n}$ uncoded submatrices, where $\Delta_B = k_B$; the rest $\ell_c = \ell - \ell_u$ submatrices are coded, i.e., 
\begin{align*}
\frac{\ell_c}{\ell} = 1 - \frac{\ell_u}{\ell} = 1 - \frac{\Delta_A \Delta_B/n}{\Delta_A/k_A} = 1 - \frac{k_A k_B}{n} = \frac{s_m}{n}. 
\end{align*}The usual assumption is that $s_m \ll n$. This indicates that a small portion of the whole storage capacity for $\bfA$ is allocated for the coded submatrices. Thus, the worker nodes will take less time to compute their assigned block-products in our proposed approach. We clarify this with an example below.

Consider that $\bfA \in \mathbb{R}^{t \times r}$ and $\bfB \in \mathbb{R}^{t \times w}$ are two sparse random matrices, where the entries are chosen independently to be non-zero with probability $\eta$. Thus, when we obtain a coded submatrix as the linear combination of $k_A$ submatrices of $\bfA$, the probability of any entry to be non-zero is approximately $k_A \eta$ (here we assume $\eta$ is small). Similarly, the probability of any entry in a coded submatrix of $\bfB$ to be non-zero is approximately $k_B \eta$, if it is obtained by a linear combination of $k_B$ submatrices. Now for the dense coded approaches \cite{yu2017polynomial, 8849468, 8919859}, every worker node stores $1/k_A$ and $1/k_B$ fractions of matrices $\bfA$ and $\bfB$, and thus the computational complexity of every worker node is approximately $O \left( (\eta k_A \eta k_B \times t ) \times \frac{r}{k_A} \frac{w}{k_B} \right) = O \left( \eta^2 \times rwt \right)$.

In our proposed approach with $x = 0$, the fraction of uncoded $\bfA$ submatrices is $\frac{k_A k_B}{n}$ and the remaining $\frac{s_m}{n}$ fraction is coded and obtained from linear combination of $k_A$ submatrices. Moreover, the coded submatrix for $\bfB$ is obtained by a random linear combination of $\zeta$ uncoded submatrices. Thus, the computational complexity for a worker node to compute the block product between an uncoded $\bfA$ and coded $\bfB$ submatrix is $O \left( (\eta  \times \eta\zeta \times t) \frac{r}{\Delta_A} \frac{w}{k_B} \right) = O \left( \eta^2 \times rwt \times \frac{\zeta}{\Delta_A \Delta_B} \right)$. Similarly, the computational complexity for a worker node to compute the block product between a coded $\bfA$ and coded $\bfB$ submatrix is $O \left( (\eta k_A \times \eta \zeta \times t) \frac{r}{\Delta_A} \frac{w}{k_B} \right) = O \left( \eta^2 \times rwt \times \frac{\zeta k_A}{\Delta_A \Delta_B} \right)$. Since the workers need to compute $p$ uncoded-coded  and $\ell - p$ coded-coded block products, the total computational complexity for every worker node in our approach is approximately 
\begin{align*}
& p  \times  O \left( \frac{\eta^2 \times rwt \times \zeta}{\Delta_A \Delta_B} \right) + (\ell - p) \times  O \left( \frac{\eta^2 \times rwt \times \zeta k_A}{\Delta_A \Delta_B} \right) \\
=  & O\left( \eta^2 \times rwt \times \left( \frac{\zeta}{n} + \frac{\zeta s_m}{n k_B}\right)\right) .
\end{align*}Thus, the computational complexity of every worker node of our approach is around $O \left(\frac{\zeta}{n} \left( 1 + \frac{s_m}{k_B}\right)\right)$ times smaller than that of the dense coded approaches. So we claim that our proposed approach is much more suited to sparse input matrices than the dense coded approaches in \cite{yu2017polynomial, 8849468, 8919859}. For example, the worker computation speed in our proposed scheme is expected to be approximately $3\times$ faster than the dense coded approaches in case of the system in Fig. \ref{new_matmat}; where $n = 12, k_B = 3, \zeta = 2$ and $s_m = 3$. The worker computation speed can be further improved in our proposed approach if we consider $x > 0$, when we combine $k_A - y$ submatrices to obtain the coded submatrices of $\bfA$.

It should be noted that there are certain approaches \cite{8919859}, \cite{das2019random} where there are some ``systematic'' worker nodes which are responsible for computing only the uncoded block-products. The computational complexity of every such worker node is approximately $O \left( (\eta \times \eta  \times t ) \times \frac{r}{k_A} \frac{w}{k_B} \right) = O \left( \eta^2 \times \frac{rwt}{k_Ak_B} \right)$, which is certainly lesser than that of any worker node in our scheme. However, 
in the approaches of \cite{8919859} and \cite{das2019random}, there are $\tau$ systematic worker nodes all of which are assigned uncoded submatrices and $s_m$ parity worker nodes all of which are assigned dense encoded submatrices. It can often be the case that one of the systematic worker nodes is a full straggler (failure). In that case, the master node requires results from at least one parity worker node, where all the assigned encoded submatrices are dense. That parity worker node will require much more time to complete its job, hence the overall computation time will be higher. 

On the other hand, in our scheme we have assigned $\ell_u$ uncoded and $\ell_c$ coded $\bfA$ submatrices to each worker node, thus unlike [8] or [10], there is an underlying symmetry in our scheme so that every non-straggling worker node takes similar amount of computation time. Thus the overall computation will be faster, as demonstrated in Fig. \ref{over_comp_sparse}. 

\begin{remark}
Our scheme is also applicable for distributed matrix-vector multiplication. In that case, the usual assumption is that each worker can store the whole vector $\bfx$, and we can prove similar theorems by substituting $\gamma_B = 1$ (or $k_B = 1$).
\end{remark}

\section{Numerical Experiments and Comparisons}
\label{sec:numexp}

In this section, we compare the performance of our approach with competing methods \cite{yu2017polynomial}, \cite{8849468}, \cite{8919859}, \cite{das2019random}, \cite{das2020coded} via exhaustive numerical experiments in {\tt AWS} (Amazon Web Services) cluster where a {\tt t2.2xlarge} machine is used as the central node and {\tt t2.small} machines as the worker nodes. All the corresponding software codes related to the numerical experiments have been made publicly available \cite{anindyacode3} for ensuring reproducibility of the results.

{\bf Worker Computation Time:} We consider a distributed matrix multiplication system with $n = 24$ workers, each of which can store $\gamma_A = \frac{1}{4}$ fraction of matrix $\bfA$ and $\gamma_B = \frac{1}{5}$ fraction of matrix $\bfB$. The input matrices $\bfA$ and $\bfB$, of sizes $12000 \times 15000$ and $12000 \times 13500$, are assumed to be sparse. We assume three different cases where sparsity ($\mu$) of the input matrices are $95\%$, $98\%$ and $99\%$, respectively, which indicates that randomly chosen $95\%$, $98\%$ and $99\%$ entries of both of matrices $\bfA$ and $\bfB$ are zero. Table \ref{worker_comp} shows the corresponding comparison of the different methods for the worker computation time for this example. We note that in real world problems, it is common that the corresponding data matrices exhibit this level of sparsity (examples can be found in \cite{sparsematrices}). 

It should be noted that we have mentioned the average value of worker computation time over all the workers in case of \cite{yu2017polynomial}, \cite{8849468}, \cite{das2020coded} and our proposed approach. However, in case of \cite{8919859} and \cite{das2019random}, we have shown the average worker computation time over the parity workers only because the remaining worker nodes are message workers where there is no coding involved.

\begin{table}[t]
\caption{{\small Comparison of worker computation time (in seconds) for matrix-matrix multiplication for $n = 24, \gamma_A = \frac{1}{4}$ and $\gamma_B = \frac{1}{5}$ (*for \cite{das2019random}, we assume $\gamma_A = \frac{2}{5}$ and $\gamma_B = \frac{1}{4}$) when randomly chosen $95\%$, $98\%$ and $99\%$ entries of both of matrices $\bfA$ and $\bfB$ are zero.}}
\label{worker_comp}
\begin{center}
\begin{small}
\begin{sc}
\begin{tabular}{c c c c c}
\hline
\toprule
\multirow{2}{*}{Methods} & \multirow{2}{*}{\, s \,} & \multicolumn{3}{c}{Worker Comp. Time (s)}  \\ \cline{3-5}
&  & $\mu = 99\%$ &  $\mu= 98\%$ & $\mu = 95\%$    \\
 \midrule
Poly Code  \cite{yu2017polynomial} & $4$ & $1.23$ &  $3.10$ & $8.21$ \\
Ortho Poly \cite{8849468}    & $4$ & $1.25$ & $3.13$ & $8.14$ \\
RKRP Code \cite{8919859}   & $4$ & $1.21$ & $3.09$ & $8.10$ \\
Conv. Code* \cite{das2019random} & $4$ & $1.92$ & $5.07$  & $10.72$ \\
SCS Opt. Sch. \cite{das2020coded} & $4$ & $0.91$ &  $1.89$ & $4.67$  \\
{\textbf{Prop. Sch.}} ($x = 0$) & $4$ & $0.54$ & $0.97$  & $3.68$ \\
{\textbf{Prop. Sch.}} ($x = 2$) & $2$ & $0.45$ & $0.81$  & $3.21$ \\
\bottomrule
\end{tabular}
\end{sc}
\end{small}
\end{center}
\end{table}%

It can easily be verified from the table that the workers take significantly less time to compute the submatrix products for our proposed approach than the other methods \cite{yu2017polynomial, 8849468, 8919859, das2019random}. This is because in the other methods the coded submatrices are linear combinations of all $k_A = 4$ submatrices from $\bfA$ (or $k_B = 5$ submatrices from $\bfB$). 

The work most closely related to our approach is our prior work in \cite{das2020coded} (SCS optimal scheme, see Section V in \cite{das2020coded}). Both approaches partition $\bfA$ and $\bfB$ into $\Delta_A = \textrm{LCM}(n,k_A)$ and $\Delta_B = k_B$ submatrices, respectively. Moreover, both approaches assign some uncoded submatrices of $\bfA$ and then some coded submatrices of $\bfA$; and assign a coded submatrix of $\bfB$ to each of the worker nodes.

However, there are some crucial differences. \cite{das2020coded} requires the weight of the encoding of the $\bfA$ submatrices to be $\Delta_A - \ell_u$ which is much higher than $k_A - y$. Furthermore \cite{das2020coded} does not allow for a trade-off between the number of stragglers and the weight of the coded $\bfA$ submatrices; this is a salient feature of our approach. Moreover, for the coding of $\bfB$, SCS optimal scheme in \cite{das2020coded} assigns linear combinations of $k_B$ submatrices, whereas in our proposed approach we assign linear combinations of $\zeta$ submatrices where $\zeta$ can be significantly smaller than $k_B$. We emphasize that the our proposed approach continues to enjoy the optimal straggler resilience when $x=0$. However, we point out that we lose a small amount in the $Q/\Delta$ metric, with respect to SCS optimal scheme in \cite{das2020coded}.

\begin{figure}[t]
\centering
\resizebox{0.99\linewidth}{!}{
\begin{tikzpicture}
\begin{axis}[
width=7.0in,
height=5.03in,
at={(2.6in,0.852in)},
major x tick style = transparent,
ybar=2*\pgflinewidth,
bar width=25pt,
ymajorgrids,
xmajorgrids,
xlabel style={font=\color{white!15!black}, font = \LARGE},
ylabel style={font=\color{white!15!black}, font = \LARGE},
ylabel={Worker Computation Time (in sec)},
symbolic x coords={{\LARGE uncoded $\bfA$ - coded $\bfB$},{\LARGE coded $\bfA$ - coded $\bfB$}},
xtick = data,
scaled y ticks = false,
enlarge x limits= 0.4,
ymin=0,
ymax=1.5,
legend cell align=left,
legend image code/.code={
        \draw [#1] (-0.05cm,-0.25cm) rectangle (0.5cm,0.5cm); },
legend style={at={(0.2,0.70)}, nodes={scale=2}, anchor=south west, legend cell align=left, align=left, draw=white!15!black}
    ]
    \addplot[style={fill=mycolor2,mark=none},postaction = { pattern = vertical lines }]
            coordinates {({\LARGE uncoded $\bfA$ - coded $\bfB$}, 0.92) ({\LARGE coded $\bfA$ - coded $\bfB$}, 0.97) };
\addlegendentry{SCS optimal scheme \cite{das2020coded}}

    \addplot[style={fill=mycolor1,mark=none},postaction = { pattern = grid }]
            coordinates {({\LARGE uncoded $\bfA$ - coded $\bfB$}, 0.58) ({\LARGE coded $\bfA$ - coded $\bfB$},0.39) };
\addlegendentry{Proposed Scheme $x = 0$}

\addplot[style={fill=mycolor3,mark=none}, postaction = { pattern = horizontal lines }]
            coordinates {({\LARGE uncoded $\bfA$ - coded $\bfB$}, 0.57) ({\LARGE coded $\bfA$ - coded $\bfB$},0.24) };
\addlegendentry{Proposed Scheme $x = 2$}

    \end{axis}

\end{tikzpicture}%
}

\caption{\small The comparison of worker computation time for the case of sparse matrices with $\mu = 98\%$. We split the total time into two parts: time required for multiplying $p$ uncoded submatrices of $\bfA$ with the coded submatrix of $\bfB$ and the time required for multiplying $(\ell - p)$ coded submatrices from $\bfA$ and with the coded submatrix of $\bfB$.}
\label{fig:timuncodedcoded}
\end{figure}
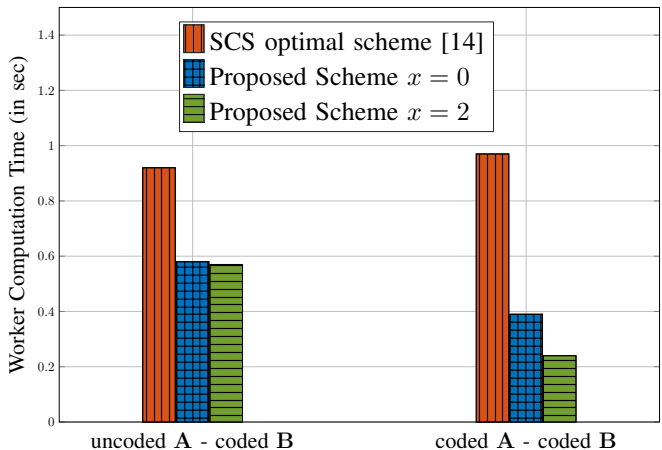

Table \ref{table:coding} shows the weight of the coding matrices for all these approaches. In this example, to obtain the coded submatrix of $\bfA$ and $\bfB$ in the SCS optimal scheme in \cite{das2020coded}, we need to have random linear combination of $19$ and $5$ submatrices, respectively; whereas our proposed method assigns linear combinations of $4$ and $3$ submatrices while having the optimal straggler resilience. Moreover, in our proposed approach, if we consider having a larger recovery threshold with $x = 2$, then the weight of coding for $\bfA$ is even lesser. Thus, the worker computation for our proposed approach is faster than the case in \cite{das2020coded} as shown in Table \ref{worker_comp}. 

This is further clarified using Fig. \ref{fig:timuncodedcoded} which shows a detailed comparison of worker computation for computing uncoded-coded and coded-coded matrix products. The proposed approach is faster than \cite{das2020coded} when computing the uncoded $\bfA$ - coded $\bfB$ products because of the reduced weight of encoded $\bfB$ submatrices. Moreover, our approach is also faster in computing the coded $\bfA$ - coded $\bfB$ products because of the reduced weight of encoded $\bfA$ submatrices. Because of these two differences in coding, worker computation is faster in our proposed approach in comparison to \cite{das2020coded}.

\begin{table}[t]
\caption{\small Weight of coding for the coded submatrices $\bfA$ and $\bfB$ for different approaches where $\Delta_A =$LCM$(n, k_A)$ and $\omega = 1 + \Bigl\lceil{\frac{s_m}{k_B}}\Bigr\rceil$. The uncoded submatrices for these approaches are assigned in the same way.}
\label{table:coding}
\begin{center}
\begin{small}
\begin{sc}
\begin{tabular}{c c c}
\hline
\toprule
Methods &  Weight of $\bfA$ & Weight of $\bfB$  \\
 \midrule
 SCS Opt. \cite{das2020coded}   &  $\Delta_A - \frac{\Delta_A k_B}{n} = 19$  & $k_B = 5$\\
 {\textbf{Prop.  ($x = 0$)}}  &  $k_A = 4$  & $1 + k_B - \Bigl\lceil{ \frac{k_B}{\omega} \Bigr\rceil} = 3$\\
 {\textbf{Prop. ($x = 2$)}}  & $k_A - \Bigl\lfloor{\frac{k_A x}{s_m}}\Bigr\rfloor = 2$ & $1 + k_B - \Bigl\lceil{ \frac{k_B}{\omega} \Bigr\rceil} = 3$\\
\bottomrule
\end{tabular}
\end{sc}
\end{small}
\end{center}
\end{table}%

{\bf Overall Computation Time:} Now we compare different approaches in terms of overall computation time to recover $\bfA^T \bfB$. Overall computation time is the time required by the worker nodes to compute the products so that the master node is able to decode all the unknowns corresponding to $\bfA^T \bfB$. Note that it is different than the worker computation time discussed above. For example, we have $24$ worker nodes, and assume that in polynomial code approach \cite{yu2017polynomial} (with recovery threshold $20$), these worker nodes require $t_0 \leq  t_1 \leq t_2 \leq \dots \leq t_{23}$ to compute their respective block-products. Then the overall computation time for this approach is $t_{19}$.

First, we consider the same system of $n = 24$ worker nodes and sparse matrices with $98\%$ or $95\%$ sparsity. We assume that there can be partial stragglers (slower workers) among these $n$ worker nodes where the slower workers have one-fifth of the speed of the non-straggling nodes. We carry out the simulations for different approaches for different number of partial stragglers. 

The results are demonstrated in Fig. \ref{over_comp_sparse} where we can see that our proposed approach is significantly faster in terms of overall computation time in comparison to the dense coded approaches for different number of slower workers. This is because our proposed approach utilizes the partial stragglers and deals with sparsity quite well. It should be noted that there are peaks for \cite{yu2017polynomial, 8849468} when there are {\it five} slower workers. The reason is that these codes are designed for {\it four} full stragglers while they do not take account the partial computations of the slower workers. 

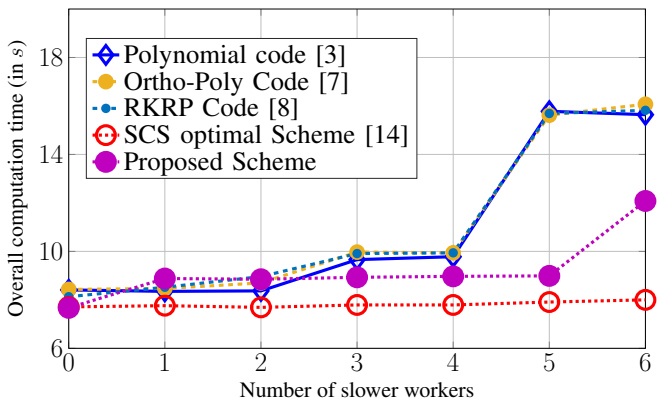
\begin{figure}[t]
\centering
\resizebox{0.99\linewidth}{!}{

\definecolor{mycolor6}{rgb}{0.92941,0.69412,0.12549}%
\definecolor{mycolor7}{rgb}{0.74902,0.00000,0.74902}%
\definecolor{mycolor8}{rgb}{0.60000,0.20000,0.00000}%

\begin{tikzpicture}
\begin{axis}[%
width=5.1in,
height=3.003in,
at={(2.6in,0.85in)},
scale only axis,
xmin=0,
xmax=6,
xlabel style={font=\color{white!15!black}, font=\Large},
xlabel={Number of slower workers},
ymin=6,
ymax=20,
ytick={6, 10, 14, 18},
xtick={0,1,2,3,4,5,6},
tick label style={font=\LARGE} ,
ylabel style={font=\color{white!15!black}, font=\Large},
ylabel={Overall computation time (in $s$)},
axis background/.style={fill=white},
xmajorgrids={true},
ymajorgrids={true},
legend style={at={(0.03,0.5)}, nodes={scale=1}, anchor=south west, legend cell align=left, align=left, draw=white!15!black,font = \LARGE}
]

\addplot [solid, color=blue, line width=2.0pt, mark=diamond, mark options={solid, blue, scale = 3}]
  table[row sep=crcr]{%
0	8.41\\
1	8.35\\
2	8.37\\
3	9.66\\
4	9.78\\
5	15.78\\
6	15.64\\
};
\addlegendentry{Polynomial code  \cite{yu2017polynomial}}

\addplot [dashed, color=mycolor6, line width=2.0pt, mark=*, mark options={solid, mycolor6, scale = 2}]
  table[row sep=crcr]{%
0	8.43\\
1	8.47\\
2	8.70\\
3	9.97\\
4	9.94\\
5	15.62\\
6	16.07\\
};
\addlegendentry{Ortho-Poly Code \cite{8849468}}

\addplot [dashed, color=mycolor1, line width=2.0pt, mark=*, mark options={solid, mycolor1, scale = 1}]
  table[row sep=crcr]{%
0	8.13\\
1	8.52\\
2	8.96\\
3	9.91\\
4	9.94\\
5	15.69\\
6	15.82\\
};
\addlegendentry{RKRP Code\cite{8919859}}

\addplot [dotted, color=red, line width=2.0pt, mark=o, mark options={solid, red, scale = 3}]
  table[row sep=crcr]{%
0	7.71\\
1	7.76\\
2	7.69\\
3	7.79\\
4	7.79\\
5	7.91\\
6	8.00\\
};
\addlegendentry{SCS optimal Scheme \cite{das2020coded}}

\addplot [dotted, color=mycolor7, line width=2.0pt, mark=*, mark options={solid, mycolor7, scale = 3}]
  table[row sep=crcr]{%
0	7.67\\
1	8.88\\
2   8.85\\
3	8.93\\
4	8.97\\
5	8.99\\
6	12.08\\
};
\addlegendentry{Proposed Scheme}

\end{axis}
\end{tikzpicture}%
}

\caption{\small Comparison among different coded approaches in terms of overall computation time for different number of slower worker nodes when the ``input'' matrices are fully dense. The system has $n = 24$ worker nodes each of which can store $\gamma_A = 1/4$ and $\gamma_B = 1/5$ fraction of matrices $\bfA$ and $\bfB$, respectively, so the recovery threshold, $\tau = 20$. The slower workers are simulated in such a way so that they have half of the speed of the non-straggling workers.}
\label{over_comp_dense}
\end{figure}

Next we simulate another example with dense ``input'' matrices in the same system having $n = 24$ worker nodes which include some slower worker nodes which have half of the speed of the non-straggling nodes. The result for overall computation time is demonstrated in Fig. \ref{over_comp_dense}. Although, the approach in \cite{das2020coded} is slightly faster than our proposed approach (since it has a slightly smaller $Q/\Delta$ than ours), our proposed approach still outperforms the dense coded approaches in \cite{yu2017polynomial, 8849468, 8919859} in most of the cases.

Now in order to clarify the trade-off between $Q/\Delta$ and the sparsity of the code, we emphasize that there are two fold gains in computation time in our scheme: utilization of partial stragglers and dealing with sparse matrices. Because of these two fold gains, when the ``input'' matrices are sparse, our approach outperforms all the approaches in \cite{yu2017polynomial, 8849468, 8919859} and \cite{das2020coded} as shown in Fig. \ref{over_comp_sparse}. On the other hand, if the ``input'' matrices are dense, we have the advantage in computation time because of  utilization of partial stragglers only. In that case, since the approach in \cite{das2020coded} has a slightly smaller $Q/\Delta$ than ours, its overall computation is slightly faster, (see Fig. \ref{over_comp_dense}). 

{\bf Value of $Q/\Delta$:} Many of the available approaches in coded matrix computations literature \cite{yu2017polynomial, 8849468, 8919859, das2019random} cannot leverage the slow workers, because they assign exactly one job to each of the worker nodes. On the other hand, the proposed approach and \cite{das2020coded} assign multiple jobs to each of the worker nodes. This allows the opportunity to leverage partial stragglers. 

We consider the same example of $n = 24$ worker nodes and show the comparison in Table \ref{table:qkappa} in terms of $Q/\Delta$. We can see that our approach has a slightly higher $Q/\Delta$ than the approach in \cite{das2020coded} and the value of $Q/\Delta$ can increase for the choice of $x > 0$. However our proposed approach has the same straggler resilience as \cite{das2020coded} and has a significant gain over that approach in terms of worker computation speed as shown in Table \ref{worker_comp}. 

It should be noted that the approaches in \cite{yu2017polynomial, 8849468, 8919859} can directly be extended to utilizing the partial computations of the stragglers, and the corresponding $Q/\Delta$ value can be improved to $1$. However, in that case, the size of corresponding system matrix increases significantly and that can lead to a very high worst case condition number. It indicates that the systems can be numerically unstable and the recovered results can be quite inaccurate. This is also discussed in Table IX in \cite{das2020coded} with numerical examples. 

\begin{table}[t]
\caption{\small Comparison of utilization of partial stragglers and numerical stability among different approaches}
\label{table:qkappa}
\begin{center}
\begin{small}
\begin{sc}
\begin{tabular}{c c c}
\hline
\toprule
Methods &  $\;Q/\Delta$ & $\kappa_{worst}$   \\
 \midrule
\; \; Poly Code  \cite{yu2017polynomial} \; \; & \; \; $N/A$\; \;  & $2.40 \times 10^{10}$\\
\; \; Ortho-Poly \cite{8849468} \; \;    & \; \; $N/A$\; \;  & $1.96 \times 10^6$\\
\; \; RKRP Code\cite{8919859}\; \;    & \; \; $N/A$\; \;  & $2.83 \times 10^5$\\
\; \; Conv Code* \cite{das2019random}\; \;   & \; \; $N/A$\; \;  & $2.65 \times 10^4$\\
\; \; SCS Opt. Sch. \cite{das2020coded}\; \;   & \; \; $124/120$\; \;  & $4.93 \times 10^6$\\
 {\textbf{Prop. Sch. ($x = 0$)}}  & $139/120$  & $2.37 \times 10^6$\\
 {\textbf{Prop. Sch. ($x = 2$)}}  &  $\frac{141}{120} \leq \frac{Q}{\Delta} \leq \frac{142}{120}$\; \;  & $2.25 \times 10^4$\\
\bottomrule
\end{tabular}
\end{sc}
\end{small}
\end{center}
\end{table}%

 {\bf Numerical Stability:} For the same system we find the worst case condition number ($\kappa_{worst}$) of the decoding matrices over all different choices of $s$ stragglers for all the approaches and present them in Table \ref{table:qkappa}. As expected, the polynomial code approach \cite{yu2017polynomial} has a very high $\kappa_{worst}$. The works in \cite{8849468, 8919859, das2019random} have significantly smaller $\kappa_{worst}$; however they cannot leverage the partial computations done by the slower worker nodes. The method in \cite{das2020coded} can utilize the partial stragglers and also provide $\kappa_{worst}$ in the similar range in comparison to \cite{8849395}.

Now from Table \ref{table:qkappa}, we can see that our approach for $x = 0$ also provides comparable values for worst case condition numbers, which indicates that the corresponding decoding matrix will not be ill-conditioned. The value of $\kappa_{worst}$ is further reduced when we consider the case of $x = 2$. Although for $x = 2$, we have resilience to less number of stragglers, it has significant advantage on worker computation time as we have shown in Table \ref{worker_comp}.

\section{Conclusion}
\label{sec:conclusion}

In this work, we have proposed a distributed coded matrix-multiplication scheme which (i) is resilient to optimal number of stragglers, (ii) leverages the partial computations done by the slower worker nodes and (iii) allows limited amount of coding so that the scheme is suited specifically to the case when the ``input'' matrices are sparse. 

In our scheme, most of the assigned $\bfA$ submatrices are uncoded which can preserve the sparsity of the input matrix $\bfA$. Thus, the worker computation speed is significantly faster in our method in comparison to some other prior works. Moreover, our proposed scheme also allows a trade off between the straggler resilience and the worker computation speed. Comprehensive numerical experiments on Amazon Web Services (AWS) cluster support our findings. 

There are several directions for the future work of this paper. It may be interesting to examine if multiple $\bfB$ submatrices can be assigned to a worker node to reduce the overall worker computation time while maintaining all the desirable properties of the scheme. Furthermore, we have the defined $Q$ as the worst case number of symbols to recover the final result; analysis on the average number of symbols can be of interest.

\appendix

\subsection{Proof of Claim \ref{claim:diffV}}
\label{App:proofclaim1}

\begin{proof}
We know $|\calC_m| = k_A$, and we use the random linear combinations of $k_A - y$ uncoded submatrices from matrix $\bfA$ to obtain the coded symbols. Now consider the indices of the submatrices of class $\calC_m$ using the row vector
\begin{align}
\label{eq:indices}
    \tilde{\bfz} = \begin{bmatrix} m \, & \, \ell+m \,& \, 2\ell+m  \,& \, \dots \,& \, (k_A-1)\ell+m  \end{bmatrix};
\end{align} and consider the index parameter for $\calC_m$ as $\lambda_m$. According to lines 12 and 13 of Alg. \ref{Alg:New_Optimal_Matmat}, we assign any coded submatrix of $\calC_m$ using the random linear combinations of $k_A - y$ submatrices from $\calC_m$. The indices of these submatrices are the consecutive $k_A - y$ entries of $\tilde{\bfz}$ in \eqref{eq:indices}, starting from $\lambda_m$. Next we shift the index parameter right by $k_A - y$ in a cyclic fashion. In this way, after assigning all the coded symbols corresponding to $\calC_m$, let us assume that the index parameter is $\lambda_m^{end}$. Now without loss of generality, assume that $i < j$, and we have the following three cases.

\noindent {\bf Case 1:} If $\lambda_m^{end} \leq i < j$, $|\bfV_i|= |\bfV_j|$.

\noindent {\bf Case 2:} If $i < \lambda_m^{end} \leq j$, $|\bfV_i| - |\bfV_j| = 1$.

\noindent {\bf Case 3:} If $i < j < \lambda_m^{end}$, $|\bfV_i|= |\bfV_j|$.

Thus in all three cases, we have $\Bigl||\bfV_i| - |\bfV_j|\Bigr| \leq 1$.
\end{proof}

\subsection{Proof of Lemma \ref{lem:new_opt_matmat}}
\label{App:prooflemma1}
\begin{proof}

Let $\bar{w} = \frac{\Delta_A}{n}$, so we can write $i = \delta \bar{w} + \alpha$ for any $\bfA_i$ where $0 \leq \alpha \leq \bar{w} - 1$ and $0 \leq \delta \leq n - 1$. Let $\tilde{f}(z)$ for $z \in \mathbb{Z}$ be the function defined below.
\begin{align*}
    \tilde{f}(z) = \begin{cases}
    z & \text{~if $z \geq 0$,}\\
    n + z & \text{~otherwise.}
    \end{cases}
\end{align*} where $-n < z < n$. With $p = \frac{\Delta}{n} = k_B \bar{w}$ and from the definition of $T$ in Alg. \ref{Alg:New_Optimal_Matmat}, the uncoded submatrices assigned to the worker node $W_j$ are given by $\bfA_{j\bar{w}}, \bfA_{j\bar{w}+1}, \bfA_{j\bar{w}+2}, \dots, \bfA_{j\bar{w}+p-1}$ (indices reduced modulo $\Delta_A$). 
This implies that $\bfA_i$ is assigned to worker nodes $W_{\tilde{f}(\delta)}, W_{\tilde{f}(\delta-1)}, \dots, W_{\tilde{f}(\delta - k_B + 1)}$ as an uncoded block. The function $\tilde{f}(\cdot)$ handles negative indices. So, we have $|\bfU_i| = k_B$, which proves part (i) of the lemma.

Next, we observe that the uncoded assignment to  $W_j$, namely $\bfA_{j\bar{w}}, \bfA_{j\bar{w}+1}, \dots, \bfA_{j\bar{w}+p-1}$, belong to the classes $\calC_{j\bar{w}}, \calC_{j\bar{w}+1}, \dots, \calC_{j\bar{w}+p-1}$ and $\ell_c = \ell - p$ coded submatrices consist of the elements from $\calC_{j\bar{w} + p}, \calC_{j\bar{w}+p+1}, \dots, \calC_{j\bar{w}+\ell-1}$, respectively. In particular, it follows that the assigned submatrices to any worker node $W_j$ belong to the classes $\calC_{j\bar{w}}, \calC_{j\bar{w}+1}, \dots, \calC_{j\bar{w}+\ell-1}$,  so that each worker node contains an assignment from each class $\calC_m, 0 \leq m \leq \ell-1$. Thus, any submatrix cannot appear more than once within any particular worker whether in an uncoded or a coded block. This shows that $\bfU_i \cap \bfV_i = \emptyset$.

Now, consider the $j$-th symbols (i.e., the $j$-th $\bfA$ submatrices) of the consecutive $\ell$ worker nodes  within a worker group $\calG_{\lambda}$, for $\lambda \in\{ 0, 1, \dots, c-1\}$. We can show that all the $\ell$ classes are represented over those $j$-th symbols of these corresponding $\ell$ worker nodes. To prove it by contradiction, we assume that there are two different workers, $W_{i_1}$ and $W_{i_2}$, where the corresponding symbols come from the same class. Since the assignments to $W_i$ are $\calC_{i\bar{w}}, \calC_{i\bar{w}+1}, \calC_{i\bar{w}+2}, \dots, \calC_{i\bar{w}+p-1}$, this is only possible if $\left[ (i_1 - i_2) \bar{w} \right] \,\, \textrm{mod} \, \ell = 0$. But it is not possible since $|(i_1 - i_2)| < \ell$ and $\bar{w}$ and $\ell$ are co-prime (since $\Delta_A = \textrm{LCM} (n, k_A)$).  Thus the $j$-th symbols of $\ell$ workers of $\calG_{\lambda}$ come from $\ell$ different $\calC_m$'s. Together with the fact that each worker node contains all $\calC_m$'s, we can conclude part (iii).

Applying the above argument to the $j$-th coded symbols of $\ell$ consecutive workers within a group, we conclude that these come from $\ell$ different $\calC_m$'s. But we have $c$ such worker groups each consisting of $\ell$ consecutive workers. Moreover, every worker has $\ell_c$ coded symbols, thus there are, in total, $\ell_c \times c$ coded symbols from any class $\calC_m$. Since each of the coded symbols consists of $k_A - y$ submatrices from $\calC_m$, we have
\begin{align*}
    \sum\limits_{i: \bfA_i \in \calC_m} |\bfV_i| & = \ell_c \times \frac{n}{\ell} \times (k_A - y) \\ & = \left( \frac{\Delta_A}{k_A} - \frac{\Delta_A \Delta_B}{n}\right) \times \frac{n}{\Delta_A/k_A} \times (k_A - y) \\ & = s_m \left( k_A - y \right).
\end{align*} Thus, the average of those corresponding $|\bfV_i|$'s is given by
\begin{align}
\label{eq:avg_mu}
    \mu_m = \frac{\sum\limits_{i: \bfA_i \in \calC_m} |\bfV_i|}{k_A} = \frac{ s_m \left( k_A - y \right)}{k_A}.
\end{align} But according to Claim \ref{claim:diffV}, $\Bigl||\bfV_i| - |\bfV_j|\Bigr| \leq 1 ; $ for any $i, j$ such that $\bfA_i, \bfA_j \in \calC_m$. Thus  $\floor{\mu_m} \leq |\bfV_i| \leq \ceil{\mu_m}$, so that
\begin{align*}
    |\bfV_i| \geq \; \Bigl\lfloor{s_m - \frac{s_m y}{k_A}}\Bigr\rfloor  & = \Bigl\lfloor{s_m - \frac{s_m}{k_A} \Bigl\lfloor{\frac{k_A x}{s_m}\Bigr\rfloor}}\Bigr\rfloor \\
    & \geq \Bigl\lfloor{s_m - \frac{s_m}{k_A} {\frac{k_A x}{s_m}}}\Bigr\rfloor  = \; s_m - x \; = \; s \, .
\end{align*} This concludes the proof for part (ii).
\end{proof}

\subsection{Proof of Corollary \ref{cor:matmat}}
\label{App:proofcorr2}

\begin{proof}
Consider two elements, $\bfA_u, \bfA_v \in \calC_m$ where $u \neq v$. From Lemma \ref{lem:new_opt_matmat}, we know that $|\bfU_u| = |\bfU_v| = k_B$. Now from the proof of Lemma \ref{lem:new_opt_matmat}, we also know that each of the $\ell$ submatrices assigned to any worker node comes from a different equivalence class (indices modulo $\ell$), i.e., $\bfU_u \cap \bfU_v = \emptyset$. It can be proved similarly for any two arbitrary elements of $\calC_m$. Now, since $|\calC_m| = k_A$, $\abs*{ \left( \underset{i: \bfA_i \in \calC_m }{\cup} \bfU_i\right) } = k_A k_B$. 

Now consider any $i$ such that $\bfA_i \in \calC_m$. According to the proof of Lemma \ref{lem:new_opt_matmat}, each worker node contains exactly one assignment of $\calC_m$, i.e., $\bfU_i \cap \bfV_{j \ell+m} = \emptyset$; if $\bfA_{j \ell+m} \in \calC_m$.
Thus we have $\left( \underset{i: \bfA_i \in \calC_m }{\cup} \bfU_i\right)  \; \cap \;  { \left( \underset{i: \bfA_i \in \calC_m }{\cup} \bfV_i\right) } \; = \; \emptyset$.

But $\calC_m$ appears at every worker node, so we can say that $\abs*{\left( \underset{i: \bfA_i \in \calC_m }{\cup} \bfU_i\right)  \; \cup \;  { \left( \underset{i: \bfA_i \in \calC_m }{\cup} \bfV_i\right) }} = n$; and using the properties mentioned above, we have $\abs*{ \left( \underset{i: \bfA_i \in \calC_m }{\cup} \bfV_i\right) } = n - k_A k_B = s_m $.
\end{proof}

\subsection{Proof of Claim \ref{claim:appearB}}
\label{App:proofclaim2}
\begin{proof}
According to Definition \ref{def:typeB}, the assigned submatrix of $\bfB$ for worker node $W_k$ (minimum node index where $\bfA_i$ appears) is of type $k$. Now suppose that $\bfA_i$ continues to appear at worker nodes $W_{k+1}, W_{k+2}, \dots, W_{k+d}$ where $d \geq 0$ (either coded or uncoded). Thus, the corresponding assigned submatrices of $\bfB$ are from types $k+1, k+2, \dots, k+d$ (reduced modulo $k_B$). 

Assume $\bfA_i \in \calC_m$. Since $\bfA_i$ does not appear at $W_{k+d+1}$, this implies that the assignment corresponding to $\calC_m$ in $W_{k+d+1}$ is uncoded. This is because we have $x = 0$, so that $y =0$, i.e. all $\bfA_i \in \calC_m$ participate in a coded block. Thus, suppose $\bfA_j \in \calC_m$ appears in $W_{k+d+1}$ as an uncoded block. We point out that according to the proof of Lemma \ref{lem:new_opt_matmat}, $\bfA_j$ appears at consecutive $k_B$ worker nodes in an uncoded fashion.

This means that the next worker node that where $\bfA_i$ can potentially appear is $W_{k+d+k_B+1}$, which has a submatrix of type $k+d+k_B+1$ (mod $k_B$), which is the same as type $k+d+1$ (mod $k_B$). On the other hand, if $\bfA_i$ does not appear at $W_{k+d+k_B+1}$, another member of $\calC_m$ will appear in the next $k_B$ worker nodes. Then we would need to move to $W_{k+d+2 k_B+1}$ for the next potential appearance of $\bfA_i$, where we have the submatrix from type $k+d+ 2 k_B+1$ (mod $k_B$), which is the same as type $k+d+1$ (mod $k_B$). 

The above argument shows that after having the $\bfB$ submatrices of types $k, k+1, k+2, \dots, k+d$ (mod $k_B$) from $W_k, W_{k+1}, W_{k+2}, \dots, W_{k+d}$, we will have a submatrix of type $k+d+1$ (mod $k_B$). Applying the argument recursively, we can conclude the required result.
\end{proof}

\subsection{Proof of Lemma \ref{lem:graphtheory}}
\label{App:prooflemma2}
\begin{proof}
Let $\bfb$ represent a vector with $k_B$ unknowns so that 
\begin{align*}
    \bfb^T = \begin{bmatrix}
    b_0 \; & \; b_1 \; & \; b_2 \; & \; \dots \; & \; b_{k_B - 1}
    \end{bmatrix} .
\end{align*}Let $\bfc^T = \bfb^T \, \bfR_i$, where $\bfc$ is of length $\sigma$. In order to prove the lemma, we need to show that we can decode all $k_B$ unknowns of $\bfb$ from any $k_B$ entries of $\bfc$.

Consider a bipartite graph $\calG = \calC \cup \calB$ whose vertex set consists of two sets, $\calC$ (representing any $k_B$ entries of $\bfc$) and $\calB$ (representing the entries of $\bfb$). An edge connects $c_i$ to $b_j$ if $b_j$ participates in the corresponding equation. Thus, a columns of $\bfR_i$ corresponds to a vertex in $\calC$ and the non-zero entries of the column correspond to the edges incident on the vertex. 

In the argument below, we suppose that the random linear coefficients are indeterminates and we argue that there exists a matching in $\calG$ where all the unknowns in $\calB$ are matched. Thus, according to Hall's marriage theorem \cite{marshall1986combinatorial}, we need to argue that for any $\tilde{\calC} \subset \calC$, the cardinality of the neighbourhood of $\tilde{\calC}$, denoted as $\calN (\tilde{\calC}) \subset \calB$, is at least as large as $|\tilde{\calC}|$.
%

To argue this, we partition the columns of $\bfR_i$ into $\omega$ disjoint sets, $\mathbf{\Omega}_0, \mathbf{\Omega}_1, \dots, \mathbf{\Omega}_{\omega - 1}$, where each of the sets, $\mathbf{\Omega}_0, \mathbf{\Omega}_1, \dots, \mathbf{\Omega}_{\omega - 2}$, have $k_B$ columns each, and $\mathbf{\Omega}_{\omega-1}$ has the remaining $\sigma - k_B (\omega - 1) \leq k_B$ columns.
According to Claim \ref{claim:appearB}, for any $\bfA_i$, since the corresponding $\bfB$ submatrices come from $\sigma$ consecutive types, we can partition it in such a way so that each set $\mathbf{\Omega}_0, \mathbf{\Omega}_1, \dots, \mathbf{\Omega}_{\omega - 2}$ has exactly one column corresponding to every submatrix type, and $\mathbf{\Omega}_{\omega-1}$ has the remaining ones. So, by permuting some columns of $\bfR_i$ we can equivalently state
\begin{align*}
    \bfc^T \; = \;  \bfb^T \; \; \left[ \mathbf{\Omega}_0 \; \; \; \; \mathbf{\Omega}_1 \; \; \; \; \dots \; \; \; \; \mathbf{\Omega}_{\omega-2}   \; \; \; \; \mathbf{\Omega}_{\omega-1} \right] \; 
\end{align*}where $\omega = \Bigl\lceil{ \frac{\sigma}{k_B} \Bigr\rceil} = 1 +  \Bigl\lceil{ \frac{s_m}{k_B} \Bigr\rceil}$. 

Note that the non-zero entries in the columns in any $\mathbf{\Omega}_k$ are cyclically shifted. We are trying to determine the neighborhood of any $m$ elements of $\calC$. Using the arguments in the proof of Theorems 7 and 8 in \cite{das2020coded}, we can show that for $\tilde{\calC}$, such that $|\tilde{\calC}| = m$, we have 
\begin{align*}
    |\calN (\tilde{\calC}) | = \min\; \left(\zeta + \Bigl\lceil{\frac{m}{\omega}\Bigr\rceil} - 1, k_B \right) .
\end{align*}However, we need $|\calN (\tilde{\calC}) | \geq m$; which indicates that 
\begin{align*}
    \zeta \geq 1 + m - \Bigl\lceil{\frac{m}{\omega}\Bigr\rceil} \; ; \; \textrm{for $m = 1, 2, \dots, k_B$} .
\end{align*}Since $s_m$ and $k_B$ are given, $\omega = 1 + \Bigl\lceil{\frac{s_m}{k_B}\Bigr\rceil}$ is constant for any given parameters. Now $1 + m - \Bigl\lceil{\frac{m}{\omega}\Bigr\rceil}$ is an increasing sequence for integer $m \leq k_B$, so we need to set $\zeta \geq 1 + k_B - \Bigl\lceil{ \frac{k_B}{\omega} \Bigr\rceil}$. Thus for any $\tilde{\calC}$, we have  $|\calN (\tilde{\calC})| \geq |\tilde{\calC}|$, which indicates that we have a matching. Since the entries are chosen randomly from a continuous distribution, we can say that any $k_B \times k_B$ submatrix of $\bfR_i$ is full rank.
\end{proof}

\subsection{Proof of Lemma \ref{lem:Gmmproof}}
\label{sec:Gmmproof}

\begin{proof}
The matrix $\bfG \, = \, \bfG_{A} \, \odot \, \bfG_{B}$ has a size $k_A k_B \times n$. We can pick any arbitrary $\tau$ columns of $\bfG$ which provides us with a $k_A k_B \times \tau$ submatrix, $\bfG_{sub}$. Let us choose the corresponding columns from $\bfG_A$ and $\bfG_B$ to form $G_{A_{sub}}$ and $G_{B_{sub}}$, respectively, so that $\bfG_{sub}\, = \, \bfG_{A_{sub}} \odot \bfG_{B_{sub}}$.

First we partition $\bfG_A$ into $k_A+1$ block columns, and denote them as $b_0, b_1, \dots, b_{k_A}$, where each of the first $k_A$ block columns have $k_B$ columns each and $b_{k_A}$ has $s_m$ columns. We denote $e_i$  (for $i = 0, 1, \dots, k_A - 1$) as a vector of length $k_A$, whose $i$-th element is $1$ and other elements are $0$. Thus for $i = 0, 1, \dots, k_A - 1$, we can say that the columns in $b_i$ are $e_i$. We assume that $\delta_i$ is the number of columns which are missing from the block column $b_i$ of $\bfG_A$ to form $\bfG_{A_{sub}}$, so $\sum_{i=0}^{k_A-1} \delta_i \leq s$. Thus we have $ \eta = \sum_{i=0}^{k_A-1} \delta_i + x$ columns in $\bfG_{A_{sub}}$ from $b_{k_A}$ of $\bfG_A$. 

Next from Lemma \ref{lem:new_opt_matmat}, part (ii), we have at least $s = s_m - x$ appearances of any $\bfA_i \in \calC_m$ in block $b_{k_A}$ in $\bfG_A$. Thus, out of those chosen $\eta$ columns in block $b_{k_A}$, every such $\bfA_i$ will appear in at least $\eta_0 = \eta - x = \sum_{i=0}^{k_A-1} \delta_i$ columns. 

Now, let us choose the columns of $\bfG_{A_{sub}}$ which are from $b_0$.  Next choose $\delta_0 \leq \eta_0$ such columns of $\bfG_{A_{sub}}$ which are from $b_{k_A}$ of $\bfG_A$ having non-zero entries at index $0$. We set each of those $\delta_0$ columns as $e_0$. Thus we have $k_B$ columns each of which is $e_0$, and after the Khatri-Rao Product with the corresponding $k_B$ columns from $\bfG_{B_{sub}}$, we have a matrix having the following form
\begin{align*}
\bfu_0 \, = \underbrace{\begin{bmatrix}
\, \bfR^0_{k_B,k_B} & \textbf{0}_{k_B,k_B} & \textbf{0}_{k_B,k_B} & \dots & \textbf{0}_{k_B,k_B}
\end{bmatrix}^T}_{k_A \, \textrm{blocks}} 
\end{align*} where $\bfR^0_{k_B, k_B}$ is obtained by taking some $k_B$ columns from $\sigma = k_B + s_m$ columns corresponding $\bfA_{m_0}$ (which is the first member of $\calC_m$), for the case $x = 0$, as described in Claim \ref{claim:appearB} and Lemma \ref{lem:graphtheory}. So, $\bfR^0_{k_B, k_B}$ is full rank.

Similarly, we can choose the columns of $\bfG_{A_{sub}}$ which are from $b_1$.  Next choose $\delta_1 \leq \eta_0 - \delta_0$ more columns of $\bfG_{A_{sub}}$ which are from $b_{k_A}$ of $\bfG_A$  having non-zero entries at index $1$. We set each of those $\delta_1$ columns as $e_1$. Thus we have $k_B$ columns each of which is $e_1$, and after the Khatri-Rao Product with the corresponding $k_B$ columns from $\bfG_{B_{sub}}$, we have a matrix having the following form
\begin{align*}
\bfu_1 \, = \underbrace{\begin{bmatrix}
 \, \textbf{0}_{k_B,k_B} & \bfR^1_{k_B,k_B} & \textbf{0}_{k_B,k_B} & \dots & \textbf{0}_{k_B,k_B}
\end{bmatrix}^T}_{k_A \, \textrm{blocks}}
\end{align*} where $\bfR^1_{k_B, k_B}$ is obtained by taking some $k_B$ columns from $\sigma = k_B + s_m$ columns for $\bfA_1$, for the case $x = 0$, as described in Claim \ref{claim:appearB} and Lemma \ref{lem:graphtheory}. So, $\bfR^1_{k_B, k_B}$ is full rank.

We can continue the similar process for $b_2, b_3, \dots, b_{k_A - 1}$, and we can show that we will have enough remaining columns even for $b_{k_A-1}$, since $\eta_0 = \sum_{i=0}^{k_A-1} \delta_i$. Thus in this way, we can show that we have a $k_A \times k_A$ block diagonal matrix, $\bfu = [\bfu_0 \;\; \bfu_1  \; \; \dots \; \bfu_{k_A - 1}]$, each of whose diagonal blocks is of size $k_B \times k_B$ and of full rank, thus the whole block diagonal matrix is full rank. 

Finally, as there exists a choice of values that makes the chosen $k_A k_B \times k_A k_B$ submatrix of $\bfG_{sub}$ nonsingular, it continues to be nonsingular with probability 1 under a random choice. It should be noted that we will have some additional $x = \tau - k_A k_B$ columns in $\bfG_{sub}$, and thus any $k_A k_B \times \tau$ submatrix of $\bfG$ has a rank $k_A k_B$ with probability $1$.    
\end{proof}

\subsection{Proof of Claim \ref{claim:Qprev}}
\label{app:proofclaim3}

\begin{proof}
Consider any worker group $\calG_{\lambda}$, $\lambda =0, 1, \dots, c - 1$ consisting of $\ell$ worker nodes. The submatrices corresponding to $\calC_m$ appear in all different locations, $0, 1, \dots, \ell-1$ (cf. Lemma \ref{lem:new_opt_matmat}), in $\calG_{\lambda}$. Let $\alpha_0$ be the maximum number of submatrix-products that can be processed across all worker nodes of $\calG_{\lambda}$ such that none of those submatrices from $\calC_m$ is processed even once. Then, from Fig. \ref{proof_Q} it can be seen that 
\begin{align}
\label{eq:alpha} 
\alpha_0 = 0 + 1 + 2 + \dots + \ell - 1 = \frac{\ell(\ell-1)}{2}.
\end{align}
Now we know that we can process $\alpha_0$ submatrix-products from each of the worker groups without processing any submatrix-product corresponding to $\calC_m$. Any additional processing will necessarily process a submatrix-product corresponding to $\calC_m$. Suppose we choose any particular worker, where the position index of $\calC_m$ is $i$. In that case, we can acquire at most $\ell - 1 - i$ more symbols (i.e., submatrix-products) from that particular worker without any more appearances of $\calC_m$. A corresponding example is shown in Fig. \ref{proof_Q}.  

\begin{figure}[t]
\centering
\resizebox{0.99\linewidth}{!}{
\begin{tikzpicture}[auto, thick, node distance=2cm, >=triangle 45]
\draw
    node [sum1, minimum size = 1.2 cm, fill=green!30] (blk1) {$W_0$}
    node [sum1, minimum size = 1.2 cm, fill=green!30,right = 1.2 cm of blk1] (blk2) {$W_1$}
    node [sum1, minimum size = 1.2 cm, fill=green!30,right = 1.2 cm of blk2] (blk4) {$\dots$}
    node [sum1, minimum size = 1.2 cm, fill=green!30,right = 1.2 cm of blk4] (blk5) {$W_{\ell -2}$}

    node [sum1, minimum size = 1.2 cm, fill=green!30,right = 1.2 cm of blk5] (blk01) {$W_{\ell -1}$}

    node [block, minimum width = 1.8cm, fill=green!30,below = 0.4 cm of blk1] (blk11) {$\calC_0$}
    node [block, minimum width = 1.8cm, fill=green!30,below = 0.0005 cm of blk11] (blk12) {$\calC_1$}
    node [block,  minimum width = 1.8cm, fill=green!30,below = 0.0005 cm of blk12] (blk14) {$\dots$}
    node [block,  minimum width = 1.8cm, fill=green!30,below = 0.0005 cm of blk14] (blk15) {$\calC_{\ell - 3}$}
    node [block,  minimum width = 1.8cm, fill=green!30,below = 0.0005 cm of blk15] (blk16) {$\calC_{\ell - 2}$}
    node [block,  minimum width = 1.8cm, fill=orange!50,below = 0.0005 cm of blk16] (blk17) {$\calC_{\ell - 1}$}

    node [block, minimum width = 1.8cm, fill=green!30,below = 0.4 cm of blk2] (blk21) {$\calC_1$}
    node [block, minimum width = 1.8cm, fill=green!30,below = 0.0005 cm of blk21] (blk22) {$\calC_2$}
    node [block,  minimum width = 1.8cm, fill=green!30,below = 0.0005 cm of blk22] (blk24) {$\dots$}
    node [block,  minimum width = 1.8cm, fill=green!30,below = 0.0005 cm of blk24] (blk25) {$\calC_{\ell -2}$}
    node [block,  minimum width = 1.8cm, fill=orange!50,below = 0.0005 cm of blk25] (blk26) {$\calC_{\ell -1}$}
    node [block,  minimum width = 1.8cm, fill=green!30,below = 0.0005 cm of blk26] (blk27) {$\calC_0$}


    node [block, minimum width = 1.8cm, fill=green!30,below = 0.4 cm of blk4] (blk41) {$\dots$}
    node [block, minimum width = 1.8cm, fill=green!30,below = 0.0005 cm of blk41] (blk42) {$\dots$}
    node [block,  minimum width = 1.8cm, fill=orange!50,below = 0.0005 cm of blk42] (blk44) {$\dots$}
    node [block,  minimum width = 1.8cm, fill=green!30,below = 0.0005 cm of blk44] (blk45) {$\dots$}
    node [block,  minimum width = 1.8cm, fill=green!30,below = 0.0005 cm of blk45] (blk46) {$\dots$}
    node [block,  minimum width = 1.8cm, fill=green!30,below = 0.0005 cm of blk46] (blk47) {$\dots$}

    node [block, minimum width = 1.8cm, fill=green!30,below = 0.4 cm of blk5] (blk51) {$\calC_{\ell - 2}$}
    node [block, minimum width = 1.8cm, fill=orange!50,below = 0.0005 cm of blk51] (blk52) {$\calC_{\ell - 1}$}
    node [block,  minimum width = 1.8cm, fill=green!30,below = 0.0005 cm of blk52] (blk54) {$\dots$}
    node [block,  minimum width = 1.8cm, fill=green!30,below = 0.0005 cm of blk54] (blk55) {$\calC_{\ell - 5}$}
    node [block,  minimum width = 1.8cm, fill=green!30,below = 0.0005 cm of blk55] (blk56) {$\calC_{\ell - 4}$}
    node [block,  minimum width = 1.8cm, fill=green!30,below = 0.0005 cm of blk56] (blk57) {$\calC_{\ell - 3}$}
    
    node [block, minimum width = 1.8cm, fill=orange!50,below = 0.4 cm of blk01] (blk011) {$\calC_{\ell - 1}$}
    node [block, minimum width = 1.8cm, fill=green!30,below = 0.0005 cm of blk011] (blk012) {$\calC_{0}$}
    node [block,  minimum width = 1.8cm, fill=green!30,below = 0.0005 cm of blk012] (blk014) {$\dots$}
    node [block,  minimum width = 1.8cm, fill=green!30,below = 0.0005 cm of blk014] (blk015) {$\calC_{\ell - 4}$}
    node [block,  minimum width = 1.8cm, fill=green!30,below = 0.0005 cm of blk015] (blk016) {$\calC_{\ell - 3}$}
    node [block,  minimum width = 1.8cm, fill=green!30,below = 0.0005 cm of blk016] (blk017) {$\calC_{\ell - 2}$}

    
node at (1.1, -6.7) {{\color{blue} $0$}}
node at (3.5, -6.7) {{\color{blue} $1$}}
node at (6, -6.7) {{\color{blue} $\ell - 1 - i$}}
node at (8.5, -6.7) {{\color{blue} $\ell - 2$}}
node at (10.9, -6.7) {{\color{blue} $\ell - 1$}}

    ;

\draw[->](1.1,-5.3) -- (1.1,-6.3);
\draw[->](3.5,-4.7) -- (3.5,-6.3);
\draw[->](6,-3.2) -- (6,-6.3);
\draw[->](8.5,-2.5) -- (8.5,-6.3);
\draw[->](10.9,-1.8) -- (10.9,-6.3);

\draw[->](blk1) -- node{} (blk11);
\draw[->](blk2) -- node{} (blk21);
\draw[->](blk4) -- node{} (blk41);
\draw[->](blk5) -- node{} (blk51);
\draw[->](blk01) -- node{} (blk011);

\end{tikzpicture}
}
\centering
\caption{\small We partition the given $n$ worker nodes into $c = n/\ell$ disjoint groups of worker nodes. We show a group of $\ell$ worker nodes (from $W_0$ to $W_{\ell -1}$, shown in green) where any class (without loss of generality, assume $\calC_{\ell - 1}$) appears exactly $\ell$ times, in all $\ell$ different locations. Once we acquire the symbol $\calC_{\ell - 1}$ from location $i$, we can acquire $(\ell - 1 - i)$ more symbols from that worker node without any further appearances of $\calC_{\ell - 1}$. We show the corresponding numbers for the selected worker group.}

\label{proof_Q}
\end{figure}
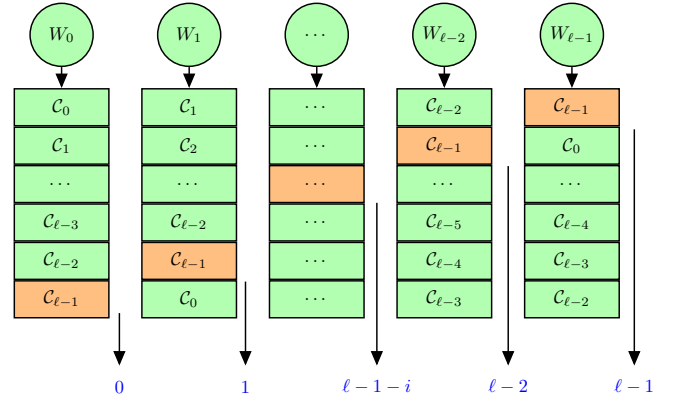 

This is true for all $c =  n/ \ell$ worker groups, since in each group any class appears at all locations. Thus, the maximum number of submatrix-products that can be processed for each additional appearance of $\calC_m$ can be expressed by the following vector.
\begin{align}
\label{eq:z}
    \bfz =\left( \undermat{c}{\ell, \ell, \dots, \ell}, \undermat{c}{\ell - 1, \dots, \ell - 1}, \dots, \undermat{c}{1, 1, \dots, 1}\right) .
\end{align}

\vspace{0.2 in}
Here $\bfz$ is a non-increasing sequence, so in order to obtain the maximum number of submatrix-products where $\calC_m$ appears at most $\kappa - 1$ times, we need to acquire submatrix-products sequentially as mentioned in $\bfz$. Let $c_1 = \floor{\frac{\kappa - 1}{c}}$ and $c_2 = \kappa - 1 - c c_1$; so $c c_1 + c_2 = \kappa - 1$. Then we can choose the first $\kappa - 1$ workers (as mentioned in $\bfz$) so that we can have $\eta$ symbols where $\calC_m$ appears exactly $\kappa - 1$ times, so 
\begin{align*}
 \eta & = c \alpha_0 +  \sum\limits_{i=0}^{\kappa - 2} \bfz[i] \\
 & = \frac{n (\ell - 1 )}{2} + c \sum\limits_{i=0}^{c_1 - 1} (\ell - i) + c_2 (\ell - c_1)  .
\end{align*} This concludes the proof.
\end{proof}

\subsection{Proof of Theorem \ref{thm:matmatq}}
\label{app:proofQ}
\begin{proof}
We assume that there is such $\bfA_i^T \bfB_j$ which can still not be recovered after some certain $Q_{ub}$ symbols are acquired, where $Q_{ub}$ is defined in the theorem statement. Assume that $\bfA_i \in \calC_m$. Now from Theorem \ref{thm:matmatstr}, we know that any $\tau = k_A k_B + x$ out of those $n$ submatrix-products corresponding to $\calC_m$ are enough to recover all the corresponding unknowns. So, to prove the upper bound for $Q$, we try to find the maximum number of block products ($Q^{'}$) which can be acquired over all the worker groups where there are at most $\tau-1$ appearances of $\calC_m$. According to  Claim \ref{claim:Qprev}, $Q'$ is given by
\begin{align*}
 Q^{'} = \frac{n (\ell - 1 )}{2} + c \sum\limits_{i=0}^{c^x_1 - 1} (\ell - i) + c^x_2 (\ell - c^x_1) ;
\end{align*}where $c^x_1 = \floor{\frac{\tau - 1}{c}}$ and $c^x_2 = \tau - 1 - c c^x_1$ with $\tau = k_A k_B + x$. It indicates that we can recover all $k_A k_B$ such unknowns corresponding to $\calC_m$ if we acquire any $Q = Q'+1$ submatrix-products over $n$ workers. Thus according to Theorem \ref{thm:matmatstr}, we can recover the unknown $\bfA_i^T \bfB_j$ from any $Q_{ub}$ submatrix-products, which concludes the proof for upper bound.

On the other hand, to prove the lower bound, we pick a particular submatrix $\bfA_{\Delta_A - 1}$ and show a certain pattern of $Q_{lb} - 1$ block-products from which $\bfA_{\Delta_A - 1}^T \bfB_j$ cannot be decoded, for any $j = 0, 1, 2, \dots, k_B - 1$. Note that $\bfA_{\Delta_A - 1} \in \calC_{\ell - 1}$ (since $\ell$ divides $\Delta_A$). 

To form that pattern of $Q_{lb} - 1$ block-products, first we choose $c \alpha_0 = \frac{n (\ell - 1)}{2}$ block-products from all $c$ worker groups where $\calC_{\ell - 1}$ does not appear. Since according to Corollary \ref{cor:matmat}(i), $\calC_{\ell - 1}$ appears at $k_A k_B$ worker nodes in an uncoded fashion, where we know that the uncoded locations of $\calC_{\ell - 1}$ in every worker group are $0, 1, \dots, p - 1$. Thus, we can acquire at most $M$ more symbols where 
\begin{align*}
    M = \left[\ell + (\ell - 1) + (\ell - 2) + \dots + (\ell - p + 1)\right] c
\end{align*} such that $\calC_{\ell - 1}$ appears only in an uncoded fashion. This is because $c \times p = k_A k_B$, where $p = \Delta/n$. However, we know that $\bfA_{\Delta_A - 1}$ appears at location $p-1$  at worker node $W_{n - k_B}$ according to the scheme.
Thus, we can construct a pattern where $\bfA_{\Delta_A - 1}$ appears exactly $k_B - 1$ times by removing $\ell - p + 1$ symbols from worker $W_{n-k_B}$. Let
\begin{align*}
    Q^{''} &=c\alpha_0 + M - (\ell-p+1)\\
    &= c \alpha_0 + \ell + (\ell - 1) + (\ell - 2)  \\
    &+ \dots + (\ell - p + 2) c + (c - 1)(\ell - p + 1) \\
    & = \frac{n (\ell - 1 )}{2} + c \sum\limits_{i=0}^{p - 2} (\ell - i) + (c  - 1) (\ell - p  + 1) \\
    & = \frac{n (\ell - 1 )}{2} + c \sum\limits_{i=0}^{c^0_1 - 1} (\ell - i) + c^0_2 (\ell - c^0_1) ;
\end{align*} since $c^0_1 = \Bigl\lfloor{\frac{k_A k_B - 1}{c} \Bigr\rfloor} = \Bigl\lfloor{\frac{k_A k_B}{c} - \frac{1}{c}\Bigr\rfloor} = p - 1$ and $c^0_2 = k_A k_B - 1 - c c_1^0 = c - 1$. Then, we have shown a pattern such that from $Q^{''}$ symbols we cannot recover the unknowns corresponding to  $\bfA_{\Delta_A - 1}$.

Now according to Claim \ref{claim:diffV} and the proof of Lemma \ref{lem:new_opt_matmat}, $\bfA_{\Delta_A - 1}$ appears at $\floor{\mu_{\ell - 1}} = s_m - \Bigl\lceil{ \frac{s_m y}{k_A}\Bigr\rceil}$ coded submatrices of $\calC_{\ell - 1}$ ($\mu_m$ is defined as the average of the coded appearances of all the submatrices of class $\calC_m$ and its value is given in \eqref{eq:avg_mu}). Since $\calC_{\ell - 1}$ appears at $s_m$ locations, it indicates that there are an additional $s_m - \floor{\mu_m} = \Bigl\lceil{ \frac{s_m y}{k_A}\Bigr\rceil}$ coded submatrices where $\bfA_{\Delta_A - 1}$ does not appear. 

Thus, we finally form a pattern of $Q^{''} + \Bigl\lceil{ \frac{s_m y}{k_A}\Bigr\rceil} = Q_{lb} - 1$ symbols where $\bfA_{\Delta_A - 1}$ appears exactly $k_B - 1$ times, where $y = \Bigl\lfloor{ \frac{k_A x}{s_m}\Bigr\rfloor}$. But there are $k_B$ unknowns corresponding to $\bfA_{\Delta_A} - 1$ in the form of $\bfA_{\Delta_A - 1}^T \bfB_j$, so we cannot decode the $\bfA_{\Delta_A - 1}^T \bfB_j$'s from this specific pattern of $Q_{lb} - 1$ symbols. This concludes the proof for the lower bound of $Q$.

\end{proof} 

\ifCLASSOPTIONcaptionsoff
  \newpage
\fi

\bibliographystyle{IEEEtran}
\bibliography{citations}
\end{document}